\documentclass[US-letter,USenglish,authorcolumns,autoref]{lipics-v2019}
\usepackage[noend]{algpseudocode}
\usepackage{algorithm}
\usepackage{times}
\usepackage{amssymb}
\usepackage{amsmath}
\usepackage{latexsym}
\usepackage{graphics}
\usepackage{url}
\usepackage{verbatim}
\usepackage{color}
\usepackage{setspace}
\usepackage{multirow}
\usepackage{lineno}
\usepackage{booktabs}
\usepackage{graphicx}
\usepackage{caption}

\algnewcommand\And{\textbf{and}}
\algnewcommand\Or{\textbf{or}}
\algnewcommand\Not{\textbf{not}}
\algnewcommand\In{\textbf{in}}
\algnewcommand\Each{\textbf{each}}

\newtheorem{conjecture}[theorem]{Conjecture}  % Numbers a conjecture by a theorem number.
\newtheorem{observation}[theorem]{Observation}      % Numbers a definition by a theorem number.

\newcommand{\squishlist}{
 \begin{list}{$\bullet$}
  { \setlength{\itemsep}{0pt}
     \setlength{\parsep}{3pt}
     \setlength{\topsep}{3pt}
     \setlength{\partopsep}{0pt}
     \setlength{\leftmargin}{2.5em}
     \setlength{\labelwidth}{1em}
     \setlength{\labelsep}{0.5em} } }

\newcommand{\squishlisttwo}{
 \begin{list}{$\triangleright$}
  { \setlength{\itemsep}{0pt}
     \setlength{\parsep}{0pt}
    \setlength{\topsep}{0pt}
    \setlength{\partopsep}{0pt}
    \setlength{\leftmargin}{2em}
    \setlength{\labelwidth}{1.5em}
    \setlength{\labelsep}{0.5em} } }

\newcommand{\squishend}{
  \end{list}  }

\usepackage{listings}

\definecolor{verbgray}{gray}{0.9}

\lstnewenvironment{code}{%
  \lstset{backgroundcolor=\color{verbgray},
 frame=single,
  language=C,
  framerule=0pt,
  basicstyle=\ttfamily,
  showstringspaces=false,
  commentstyle=\color{blue}\textit,
  keywordstyle=\color{black}\bf,
  numbers=none,
  keepspaces=true,
  columns=fullflexible}}{}

\newcommand{\same}[1]{\mathcal{S}_{#1}}
\newcommand{\opposite}[1]{\mathcal{O}_{#1}}

\newcommand{\LC}[1]{\mathcal{L}_{#1}}

\newcommand{\PCR}{\textrm{PCR}} 
\newcommand{\CCR}{\textrm{CCR}}
\newcommand{\PRR}{\textrm{PRR}}

\newcommand{\RLrep}[1]{\mathbf{{RL}}({#1})}
\newcommand{\LCrep}[1]{\mathbf{{LC}}({#1})}
\newcommand{\Samerep}[1]{\mathbf{{Same}}({#1})}

\newcommand{\RLLrep}[1]{\mathbf{{RL2}}({#1})}
\newcommand{\LCCrep}[1]{\mathbf{{LC2}}({#1})}
\newcommand{\OPPrep}[1]{\mathbf{{OPP}}({#1})}

\newcommand{\Special}{\mathbf{SP}}
\newcommand{\SpecialO}{\mathbf{SP2}}

\newcommand{\bB}{\mathbf{B}}

\newcommand{\bP}{\mathbf{P}}
\newcommand{\bC}{\mathbf{C}}

\newcommand{\bR}{\mathbf{R}}

\definecolor{shadecolor}{rgb}{.91, .91, .91}
\definecolor{bordercolor}{rgb}{.8, .8, .6}

\definecolor{ultramarine}{rgb}{0, 0.125, 0.376}

 \definecolor{arsenic}{rgb}{0.23, 0.27, 0.29}
 \definecolor{beige}{rgb}{0.96, 0.96, 0.86}

\definecolor{amber}{rgb}{1.0, 0.75, 0.0}
\definecolor{orange}{rgb}{1.0, 0.49, 0.0}
\definecolor{dandelion}{rgb}{0.94, 0.88, 0.19}

  \definecolor{indiagreen}{rgb}{0.07, 0.53, 0.03}
  \definecolor{huntergreen}{rgb}{0.21, 0.37, 0.23}

\newcommand{\blue}[1] {\textcolor{blue}{#1}}
\newcommand{\red}[1] {\textcolor{red}{#1}}

\newcommand{\bblue}[1] {{\bf \textcolor{blue}{#1}}}

\newcommand{\defo}[1] {\emph{\textcolor{blue}{#1}}}

\definecolor{shadecolor}{rgb}{.9, .9, .9}

 \usepackage{framed}

    {\endMakeFramed}

    \newenvironment{frshaded*}{%
    \MakeFramed {\advance\hsize-\width \FrameRestore}}%
    {\endMakeFramed}
    
    \newcounter{examplecounter}
\newenvironment{exam}{
 \begin{frshaded*}
    \refstepcounter{examplecounter}%
    \noindent
  \textbf{Example \arabic{examplecounter}}%
  \quad
}{%
\end{frshaded*}
}

{\endMakeFramed}

\newenvironment{frshaded2*}{%
    \MakeFramed {\advance\hsize-\width \FrameRestore}}%
    {\endMakeFramed}

\newenvironment{result}{
 \begin{frshaded2*}
}{%
\end{frshaded2*}

}

{\endMakeFramed}

\newenvironment{frshaded3*}{%
    \MakeFramed {\advance\hsize-\width \FrameRestore}}%
    {\endMakeFramed}

\usepackage[normalem]{ulem}
\usepackage{changepage}

\definecolor{winered}{rgb}{0.5,0.2,0}

\usepackage{hyperref}
\hypersetup{
    colorlinks = true,
    allcolors={winered},
}
\title{Efficient constructions of the Prefer-same and Prefer-opposite de Bruijn sequences}

\titlerunning{Efficient constructions of the Prefer-same and Prefer-opposite de Bruijn sequences}

\author{Evan Sala}{School of Computer Science, University of Guelph, Canada}{}{}{}
\author{Joe Sawada}{School of Computer Science, University of Guelph, Canada}{}{}{}
\author{Abbas Alhakim}{Department of Mathematics, American University of Beirut, Lebanon}{}{}{}

\authorrunning{E. Sala, J. Sawada and A. Alhakim} 

\Copyright{Evan Sala, Joe Sawada and Abbas Alhakim}



\nolinenumbers

\begin{document}

\maketitle

%=================================================================
%=================================================================
%=================================================================
\begin{abstract}
The greedy Prefer-same de Bruijn sequence construction was first presented by Eldert et al.~[\emph{AIEE Transactions} 77 (1958)].   As a  greedy algorithm, it has one major downside: it requires an exponential amount of space to store the length $2^n$ de Bruijn sequence.  Though de Bruijn sequences have been heavily studied over the last 60 years, finding an efficient construction for the Prefer-same de Bruijn sequence has remained a tantalizing open problem.  In this paper, we unveil the underlying structure of the Prefer-same de Bruijn sequence and solve the open problem by presenting an efficient algorithm to construct it using $O(n)$ time per bit and only $O(n)$ space.  Following a similar approach, we also present an efficient algorithm to construct the Prefer-opposite de Bruijn sequence.
\end{abstract}

\vspace{-0.2in}
%=========================================================================================
%=================================================================
%=================================================================
\section{Introduction}

Greedy algorithms often provide some of the nicest algorithms to exhaustively generate combinatorial objects, especially in terms of the simplicity of their descriptions.  An excellent discussion of such algorithms is given by Williams~\cite{williams} with examples given for a wide range of combinatorial objects
including permutations, set partitions,  binary trees,  and de Bruijn sequences.  
A downside to greedy constructions is that they generally require exponential space to keep track of which objects have already been visited.  Fortunately, most greedy constructions can also be constructed efficiently by either an iterative successor-rule approach, or by applying a recursive technique.  Such efficient constructions often provide extra underlying insight into both the combinatorial objects and the actual listing of the object being generated.

A \defo{de Bruijn sequence} of order $n$ is a sequence of bits that when considered cyclicly contains every length $n$ binary string as a substring exactly once; each such sequence has length $2^n$.   They have been studied as far back as 1894 with the work by Flye Sainte-Marie~\cite{flye}, receiving more significant attention starting in 1946  with the work of de Bruijn~\cite{DB}.   
Since then, many different de Bruijn sequence constructions have been presented in the literature (see surveys in \cite{fred-nfsr} and \cite{framework}). Generally, they fall into one of the following categories:  (i) greedy approaches (ii) iterative successor-rule based approaches which includes linear (and non-linear) feedback shift registers  (iii) string concatenation approaches  (iv) recursive approaches.  Underlying all of these algorithms is the fact that every de Bruijn sequence is in 1-1 correspondence with an Euler cycle in a related de Bruijn graph.

Perhaps the most well-known de Bruijn sequence is the one that is the lexicographically largest.  It has the following greedy Prefer-1 construction~\cite{martin}.

%----------------------------
\begin{result} \noindent  \footnotesize {  \bf Prefer-1  construction}

\smallskip

\begin{enumerate}
\item Seed with $0^{n-1}$
\item {\bf Repeat} until no new bit is added:  Append 1 if it does not create a duplicate length $n$ substring; otherwise append 0 if it does not create a duplicate length $n$ substring
\item Remove the seed 
\end{enumerate}

\vspace{-0.1in}
\end{result}
\noindent
For example, applying this construction for $n=4$ we obtain the string:
$\mbox{\sout{\red{000}}}~1111011001010000. $
Like all greedy de Bruijn sequence constructions, this algorithm has a major downside: it requires an exponential amount of space to remember which substrings have already been visited.   Fortunately, the resulting sequence can also be constructed efficiently by applying an $O(n)$ time per bit  successor-rule which requires $O(n)$ space~\cite{fred-succ}.   By applying a  necklace concatenation approach, it can even be generated in amortized $O(1)$ time per bit and $O(n)$ space~\cite{fkm2}.

Two other interesting greedy constructions take into account the last bit generated.  They are known as the Prefer-same and Prefer-opposite constructions and their resulting sequences are, respectively, the lexicographically largest and smallest with respect to a run-length encoding\footnote{The run-length encoding of a string is discussed formally in Section~\ref{sec:rle}.}~\cite{revisit}.
The Prefer-same construction  was first presented by Eldert et al.~\cite{eldert} in 1958 and was revisited with a proof of correctness by Fredricksen~\cite{fred-nfsr} in 1982.  Recently, the description of the algorithm was simplified~\cite{revisit} as follows:
%
%----------------------------
\begin{result} \noindent  \footnotesize {  \bf Prefer-same  construction}

\smallskip

\begin{enumerate}
\item Seed with length $n{-}1$ string $\cdots 01010$
\item Append 1
\item {\bf Repeat} until no new bit is added:  Append the {\bf same} bit as the last if it does not create a duplicate length $n$ substring; otherwise append the opposite bit as the last if it does not create a duplicate length $n$ substring
\item Remove the seed 
\end{enumerate}

\vspace{-0.1in}
\end{result}

\noindent
For $n=4$, the sequence generated by this Prefer-same  construction is
$ \mbox{\sout{\red{010}} }  1111000011010010.$  It has run-length encoding 44211211 which is the lexicographically largest amongst all de Bruijn sequences for $n=4$.

The Prefer-opposite construction is not greedy in the strictest sense since there is a special case when the current suffix is $1^{n-1}$.  Details about this special case are provided in the next section.  The construction presented below produces a shift of the sequence produced by the original presentation in~\cite{pref-opposite}.  Here, the initial seed of 
$0^{n-1}$ is rotated to the end so the resulting sequence is the lexicographically smallest with respect to a run-length encoding.
%
%----------------------------
\begin{result} \noindent  \footnotesize {  \bf Prefer-opposite  construction}

\smallskip

\begin{enumerate}
\item Seed with $0^{n-1}$
\item Append 0
\item {\bf Repeat} until no new bit is added:  
  \begin{itemize}
  	\item {\bf If} current suffix is $1^{n-1}$ {\bf then}: append 1 if it is the first time $1^{n-1}$ has been seen; otherwise append 0 
	\item {\bf Otherwise}: append the {\bf opposite} bit as the last if it does not create a duplicate length $n$ substring; otherwise append the same bit as the last
 \end{itemize}
\item Remove the seed 
\end{enumerate}

\vspace{-0.1in}
\end{result}

\noindent
For $n=4$, the sequence generated by this Prefer-opposite  construction is
$ \mbox{\sout{\red{000}} }  0101001101111000.$  The run-length encoding of this sequence is given by 111122143.

To simplify our discussion, let:

\begin{itemize}  
\item $\same{n} = $ the de Bruijn sequence of order $n$ generated by the Prefer-same construction,  and
\item $\opposite{n} = $ the de Bruijn sequence of order $n$ generated by the Prefer-opposite construction. 
\end{itemize}

\noindent
Unlike the Prefer-1 sequence,  and despite the vast research on de Bruijn sequences, $\same{n}$ and $\opposite{n}$ have no known efficient construction.  For $\same{n}$,  finding an efficient construction has remained an elusive open problem for over 60 years.   
The closest attempt came in 1977 when Fredricksen and Kessler  devised a construction based on lexicographic compositions~\cite{lexcomp} that we discuss further in Section~\ref{sec:LC}.

The main results of this paper are to solve these open problems by providing  successor-rule based constructions for  $\same{n}$ and $\opposite{n}$.  They generate the respective sequences in $O(n)$ time per bit using only $O(n)$ space.  %As we solve this problem, we also provide some extra insight into the sequence $\LC{n}$.   
The discovery of these efficient constructions hinged on the following idea: 
\begin{quote}
Most \emph{interesting} de Bruijn sequence are the result of joining together smaller cycles induced by \emph{simple} feedback shift registers.
\end{quote}
The initial challenge was to find such a simple underlying feedback function.  After careful study, the following function was revealed:
\[ f(w_1w_2\cdots w_n) = w_1 \oplus w_2 \oplus w_n,\]  
where $\oplus$ denotes addition modulo 2.   We demonstrate this feedback function has nice run-length properties when used to partition the set of all binary strings of length $n$ in Section~\ref{sec:PRR}.  The next challenge was to find appropriate representatives for each cycle induced by $f$ in order to apply the framework from~\cite{framework} to join the cycles together.   %This proved to be rather challenging and ultimately required insights from analyzing $\LC{n}$.

\medskip

\noindent
{\bf Outline of paper.}
Before introducing our main results, we first provide an insight into greedy constructions for de Bruijn sequences that we feel has not been properly emphasized in the recent literature.  In particular, we demonstrate how all such constructions, which are generalized by the notion of preference or look-up tables~\cite{alhakim-span,xie}, are in fact just special cases of a standard Euler cycle algorithm on the de Bruijn graph.   
This discussion is found in Section~\ref{sec:euler} which also outlines a second Euler cycle algorithm underlying the cycle joining approach applied in our main result.
In Section~\ref{sec:rle}, we present background on run-length encodings.  
 In Section~\ref{sec:feedback}, we discuss feedback functions and de Bruijn successors and introduce the function $f(w_1w_2\cdots w_n) = w_1 \oplus w_2 \oplus w_n$  critical to our main results.
  In Section~\ref{sec:generic}, we present two generic de Bruijn successors based on the framework from~\cite{framework}.
In Section~\ref{sec:same} we present our first main result: an efficient successor-rule to generate $\same{n}$.  
 In Section~\ref{sec:opposite} we present our second main result: an efficient successor-rule to generate $\opposite{n}$.  
  In Section~\ref{sec:LC} we discuss the lexicographic composition algorithm from~\cite{lexcomp} and a related open problem.
   In Section~\ref{sec:implement} we discuss implementation details and analyze the efficiency of our algorithms.
In Section~\ref{sec:proof1} and Section~\ref{sec:proof2} we detail the technical aspects required to prove our main results.
We conclude by presenting directions for future research in Section~\ref{sec:fut}.  
Implementation of our algorithms, written in C,  presented in this paper can be found in the appendices and are available for download at~\url{http://debruijnsequence.org}.

\medskip

\noindent
{\bf Applications.}
One of the first instances of de Bruijn sequences is found in works of Sanskrit prosody by the ancient mathematician Pingala dating back to the 2nd century BCE.   
Since then, de Bruijn sequences and their related theory have a rich history of application.   
One of their more prominent applications, due to their random-like properties~\cite{golomb}, is in the generation of pseudo-random bit sequences which are used in stream ciphers~\cite{cipher}.  
In particular,  linear feedback shift register constructions (that omit the string of all 0s) allow for efficient hardware embeddings which have been classically applied to represent different maps in video games including Pitfall~\cite{archeo}.  
%They also have many applications in cryptography,  in particular with respect to keystream generators~\cite{}.
Another application uses de Bruijn sequences to crack cipher locks in an efficient manner~\cite{fred-nfsr}. 
More recently, the related de Bruijn graph has been applied to genome assembly~\cite{nature,euler}.  
Given the vast literature on de Bruijn sequences and their various methods of construction, the more interesting new results may relate to sequences with specific properties.  
This makes the de Bruijn sequences $\same{n}$ and $\opposite{n}$ of special interest since they are, respectively,  the lexicographically largest and smallest sequences with respect to a run-length encoding~\cite{revisit}.  Moreover, recently it was noted they have a relatively small discrepancy, which is the maximum absolute difference between the number of 0s and 1s in any substring, when compared to the sequences generated by the Prefer-1 construction~\cite{discrep}.

%=========================================================================================
%=================================================================
%=================================================================
\section{Euler cycle algorithms and the de Bruijn graph} \label{sec:euler}

The \defo{de Bruijn graph} of order $n$ is the directed graph 
$G(n) = (V,E)$ where  $V$  is the set of all  binary  strings of length  $n$ and there is a directed edge from
$u = u_1u_2\cdots u_n$ to $v=v_1v_2\cdots v_n$ if $u_2\cdots u_n = v_1\cdots v_{n-1}$. Each edge $e$ is labeled by $v_n$.  Outputting the edge labels in a Hamilton cycle of $G(n)$ produces a de Bruijn 
sequence. Figure~\ref{fig:ham}(a) illustrates a Hamilton cycle in the de Bruijn graph $G(3)$.  Starting from 000, its corresponding de Bruijn sequence is 10111000.
%---------------------
\begin{figure}[h]
\begin{center}
\resizebox{4.3in}{!}{\includegraphics{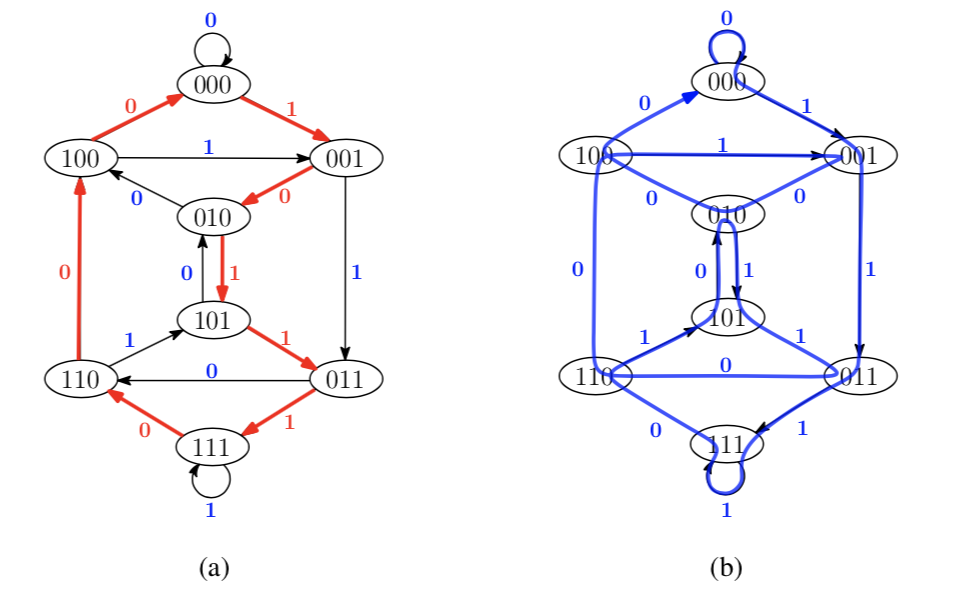}}
\end{center}

\vspace{-0.25in}
\caption{  \small
(a) A Hamilton cycle in $G(3)$ starting from 000 corresponding to the de Bruijn sequence 10111000 of order 3.
(b) An Euler cycle in $G(3)$ starting from 000 corresponding to the de Bruijn sequence 0111101011001000 of order 4.
}
\label{fig:ham}
\end{figure}
%---------------------

Each de Bruijn graph is connected and the in-degree and the out-degree of each vertex is two; the graph $G(n)$ is Eulerian.  $G(n)$ is the line graph of $G(n{-}1)$ which means an Euler cycle in $G(n{-}1)$ corresponds to a Hamilton cycle in $G(n)$. Thus, the sequence of edge labels visited in an Euler cycle is a de Bruijn sequence. Figure~\ref{fig:ham}(b) illustrates an Euler cycle in $G(3)$. The corresponding de Bruijn sequence of order four when starting from the vertex 000 is  0111101011001000.

Finding an Euler cycle in an Eulerian graph is linear-time solvable with respect to the size of the graph.  However, since the graph must be stored, applying such an algorithm to find a de Bruijn sequence requires $O(2^n)$ space.   One of the most well-known Euler cycle algorithms for directed graphs is the following due to Fleury~\cite{fleury} with details in~\cite{fred-nfsr}.  The basic idea is to not burn bridges;  in other words, do not visit (and use up) an edge if it leaves the remaining graph disconnected.

%----------------------------
\begin{result}  \small \noindent {  \bf Fleury's Euler cycle algorithm (do not burn bridges) }

\smallskip

\begin{enumerate}
\item Pick a root vertex and compute a spanning in-tree $T$
\item Make each edge of $T$ (the bridges) the last edge on the adjacency list of the corresponding vertex
\item Starting from the root, traverse edges in a depth-first manner by visiting the first unused edge in the current vertex's adjacency list
\end{enumerate}

\vspace{-0.1in}

\end{result}

\noindent
Finding a spanning in-tree $T$ can be done by reversing the direction of the edges in the Eulerian graph and computing a spanning out-tree with a standard depth first search on the resulting graph.  The corresponding edges in the original graph will be a spanning in-tree.  Using this approach, all  de Bruijn sequences can be generated by considering all possible spanning in-trees (see BEST Theorem in~\cite{fred-nfsr}). 

Although not well documented, this  algorithm is the basis for all greedy de Bruijn sequence constructions along with their generalizations using  preference tables~\cite{alhakim-span} or look-up tables~\cite{xie}.  Specifically, a preference table specifies the precise order that the edges are visited for each vertex when performing Step 3 in Fleury's Euler cycle algorithm.  Thus given a preference table and a root vertex, Step 3 in the  algorithm  can be applied to construct a de Bruijn sequence if combining the last edge from each non-root vertex forms a spanning in-tree to the root.  For example, the preference tables and corresponding spanning in-trees for the  Prefer-1 (rooted at 000), the Prefer-same (rooted at 010), and  the Prefer-opposite (rooted at 000) constructions are given in Figure~\ref{fig:trees} for $G(3)$.  For the Prefer-1, the only valid root is 000.  For the Prefer-same, either 010 or 101 could be chosen as root.   The Prefer-opposite has a small nuance.  By a strict greedy definition, the edges will not create a spanning in-tree for any root. But by changing the preference for the single string 111, a spanning in-tree is created when rooted at 000.  This accounts for the special case required in the Prefer-opposite algorithm.
Notice how these strings relate to the seeds in their respective greedy constructions.  For the Prefer-same, a root of 101 could also have been chosen, and doing so will yield the complement of the Prefer-same sequence when applying this Euler cycle algorithm.
Relationships between various preference related constructions have recently been studied in~\cite{Jiang2023}, generalizing the work in~\cite{Rubin2017} which focused on the Prefer-opposite and Prefer-1 constructions.

%---------------------
\begin{figure}[h]
\begin{center}
\resizebox{4.7in}{!}{\includegraphics{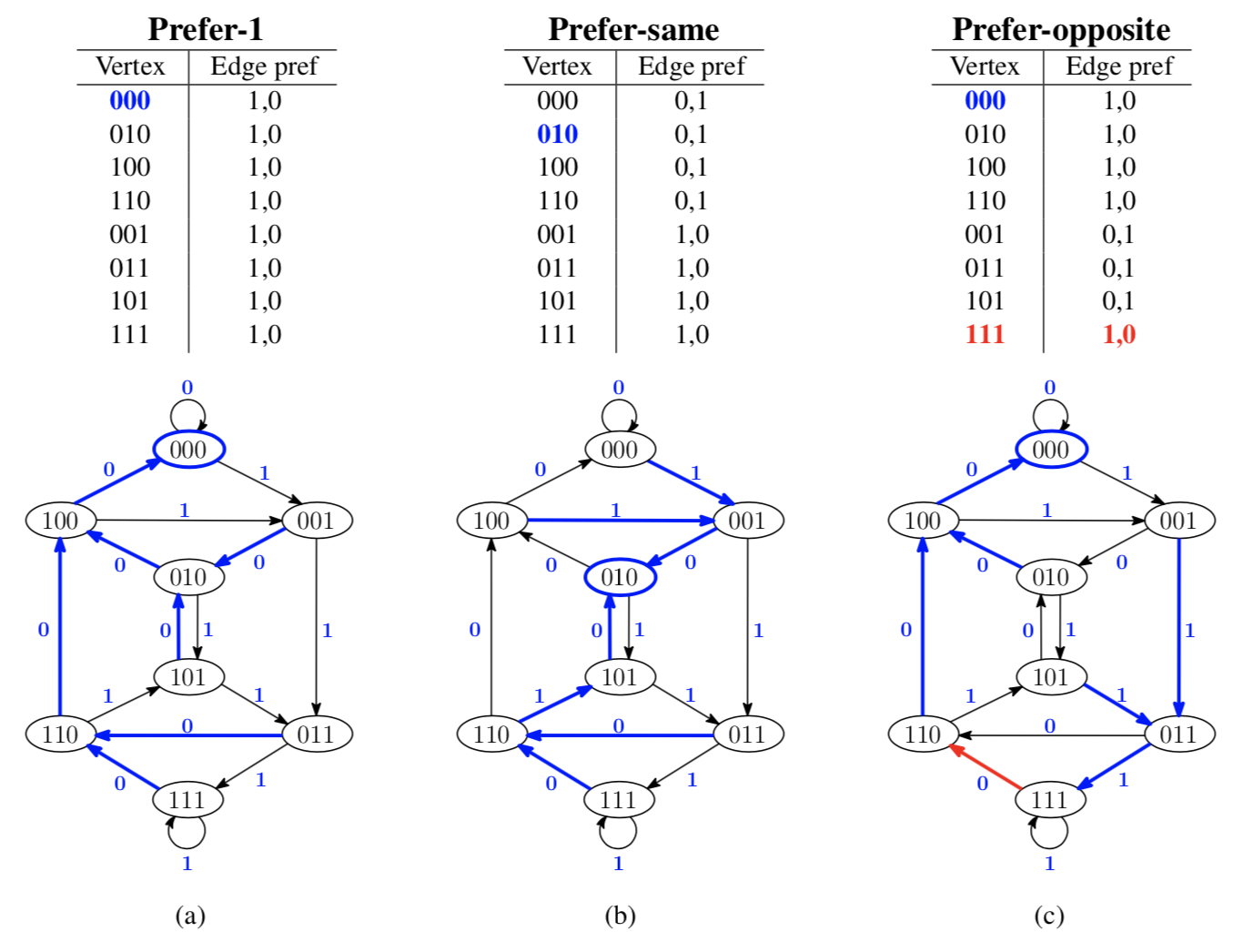}}
\end{center}

\vspace{-0.25in}
\caption{  \small
(a) A preference table corresponding to the Prefer-1 greedy construction along with its corresponding spanning in-tree rooted at 000.
(b) A preference table corresponding to the Prefer-same greedy construction along with its corresponding spanning in-tree rooted at 010.
(c) A preference table corresponding to the Prefer-opposite greedy construction along with its corresponding spanning in-tree rooted at 000.
}
 \label{fig:trees}
\end{figure}
%---------------------

A second well-known Euler cycle algorithm for directed graphs, attributed to Hierholzer~\cite{hierholzer}, is as follows: 

%----------------------------
\begin{result}  \small { \noindent  \bf Hierholzer's Euler cycle algorithm (cycle joining) }

\smallskip

\begin{enumerate}
\item Start at an arbitrary vertex $v$ visiting edges in a depth-first manner until returning to $v$, creating a cycle.
\item {\bf Repeat until all edges are visited:}  Start from any vertex $u$ on the current cycle and visit remaining edges in a DFS manner until returning to $u$, creating a new cycle.  Join the two cycles together.
\end{enumerate}

\vspace{-0.1in}

\end{result}

\noindent
This cycle-joining approach is  the basis for all successor-rule constructions of de Bruijn sequences.  A general framework for joining smaller cycles together based on an underlying feedback shift register is given for the binary case in~\cite{framework}, and then more generally for larger alphabets in~\cite{karyframework}.  It is the basis for the efficient algorithm presented in this paper, where the initial cycles are induced by a specific feedback function.

%=================================================================
%=================================================================
%=================================================================
\section{Run-length encoding} \label{sec:rle}

The sequences $\same{n}$ and $\opposite{n}$  both have properties based on a run-length encoding of binary strings.  
The \defo{run-length encoding} (RLE) of a string $\omega = w_1w_2\cdots w_n$ is a compressed representation that stores consecutively the lengths of the maximal runs of each symbol.  The \defo{run length} of $\omega$ is the length of its RLE.  
For example, the string 11000110 has RLE 2321 and run length 4.   Note that 00111001 also has RLE 2321.
Since we are dealing with binary strings, we require knowledge of  the starting symbol to obtain a given binary string from its RLE.  
As a further example:   \small
\[ \same{5}  =     11111000001110110011010001001010  \  \mbox{ has RLE } \   5531222113121111. \]  \normalsize
The following facts are proved in~\cite{revisit}. 
\begin{proposition} \label{fact:rle}
The sequence $\same{n}$  is the de Bruijn sequence of order $n$ starting with 1 that has the lexicographically largest RLE.
\end{proposition}

\begin{proposition} \label{fact:rle2}
The sequence $\opposite{n}$  is the de Bruijn sequence of order $n$ starting with 1 that has the lexicographically smallest RLE.
\end{proposition}

\noindent
Let $alt(n)$ denote the alternating sequence of 0s and 1s of length $n$ that ends with 0:  For example, $alt(6) = 101010$.  The following facts  are also immediate from~\cite{revisit}.

\begin{proposition} \label{fact:Spre}
$\same{n}$ has prefix $1^n$ and has suffix $alt(n{-}1)$.
\end{proposition}

\begin{proposition} \label{fact:Opre}
$\opposite{n}$ has length $n$ prefix $010101\cdots$ and has suffix $10^{n-1}$.
\end{proposition}
The  sequence based on lexicographic compositions~\cite{lexcomp} also has run-length properties: it is constructed by concatenating lexicographic compositions which are represented using a RLE.  Further discussion of this sequence is provided in Section~\ref{sec:LC}.

%=================================================================
%=================================================================
%=================================================================
\section{Feedback functions and de Bruijn successors} \label{sec:feedback}

Let $\bB(n)$ denote the set of all binary strings of length $n$.
We call a function $f:  \bB(n) \rightarrow \{0,1\}$ a \defo{feedback function}.  
% \red{Talk about the 4 simple functions from Golomb?}
Let $\omega = w_1w_2\cdots w_n$ be a string in $\bB(n)$.
A \defo{feedback shift register} 
is a function $F: \bB(n) \rightarrow \bB(n)$ that takes the form $F(\omega)  = w_2w_3\cdots w_n f(w_1w_2\cdots w_n)$ for a given feedback function $f$.  

A feedback function  $g: \bB(n) \rightarrow \{0,1\}$ is a  \defo{de Bruijn successor} if there exists a de Bruijn sequence of order $n$ 
such that each substring $\omega \in \bB(n)$ is followed by $g(\omega)$ in the given de Bruijn sequence.    Given a de Bruijn successor $g$ and a seed string $\omega = w_1w_2\cdots w_n$, the following function {\sc DB}($g, \omega$)  will return a de Bruijn sequence of order $n$ with suffix $\omega$:

\smallskip
\begin{result}
\vspace{-0.1in}
%=============
 \begin{algorithmic} [1]                   
\Function{DB}{$g, \omega$} 

    \For{$i\gets 1$ {\bf to}  $2^n$}  
    	\State $x_i \gets g(\omega)$
	\State $\omega \gets w_2w_3\cdots w_n x_i$
    \EndFor
    \State \Return $x_1x_2\cdots x_{2^n}$
\EndFunction
\end{algorithmic}
\vspace{-0.1in}

   \end{result}
%=============

\noindent
A \defo{linearized de Bruijn sequence} is a linear string that contains every string in $\bB(n)$ as a substring exactly once.  Such a string has length $2^n+n-1$.  
Note that the length $n$ suffix of a de Bruijn sequence $\mathcal{D}_n =$ {\sc DB}($g, w_1\cdots w_n$) is $w_1\cdots w_n$.  Thus, $w_2\cdots w_n   \mathcal{D}_n$ is a  linearized de Bruijn sequence.

%The goal of this paper is to present efficient de Bruijn successors for $\same{n}$ and $\opposite{n}$.
For each of the upcoming feedback functions, selecting appropriate representatives for the cycles they induce is an important step to developing efficient de Bruijn successors for $\same{n}$ and $\opposite{n}$.  In particular, consider two representatives for a given cycle based on their RLE.  
\begin{itemize}
\item \defo{RL-rep}:   The string with the lexicographically largest RLE; if there are two such strings, it is the one beginning with 1.  
\item \defo{RL2-rep}:  The string with the lexicographically smallest RLE;  if there are two such strings, it is the one beginning with 0.  
\end{itemize}
For our  upcoming discussion, define the \defo{period}~\footnote{The notion of a \emph{period} often allows a fractional exponent $j$, but here it must be an integer.}  of a string $\omega = w_1w_2\cdots w_n$ to be the
smallest integer $p$ such that $\omega = (w_1\cdots w_p)^j$ for some integer $j$.  If $j > 1$ we say that $\omega$ is \defo{periodic}; otherwise, we say it is \defo{aperiodic} (or primitive).

%=================================================================
\subsection{The pure cycling register (PCR)}

The \defo{pure cycling register}, denoted $\PCR$,  is the feedback shift  register with the feedback function $f(\omega) = w_1$.  Thus, $\PCR(w_1w_2\cdots w_n) = w_2\cdots w_nw_1$.  It is well-known that the $\PCR$ partitions $\bB(n)$ into cycles of strings that are equivalent under rotation.    The following example illustrates the cycles induced by the $\PCR$ for $n=5$ along with their corresponding RL-reps and RL2-reps.
%The cycle composed of $1100, 1001, 0011, 0110$ has two strings with RLE of 22 and two strings with RLE 121.  By definition,  1100 is the RL-rep and 0110 is the RL2-rep.  As a further example, consider the cycles for $n=5$.

\begin{exam}  \small \label{exam:PCR}
The $\PCR$ partitions $\bB(5)$ into the following eight cycles $\bP_1, \bP_2, \ldots, \bP_8$  where the top string in bold is the RL-rep for the given cycle.  The underlined string is the RL2-rep.

\vspace{-0.1in}

\begin{center} \footnotesize
\begin{tabular}{c @{\hskip 0.28in}  c @{\hskip 0.28in}  c @{\hskip 0.28in} c @{\hskip 0.28in} c @{\hskip 0.28in} c @{\hskip 0.28in} c @{\hskip 0.28in} c}
$\bP_1$  & 
$\bP_2$  & 
$\bP_3$  & 
$\bP_4$ & 
$\bP_5$  & 
$\bP_6$ & 
$\bP_7$ & 
$\bP_8$  \\ 
				\bblue{11010} 		&	\bblue{00101}   			& 	\bblue{11110}			&	\bblue{00001}	&	\bblue{11100}				&	\bblue{00011}	&	\underline{\bblue{11111}}		& \underline{\bblue{00000}}							 \\
				 \underline{10101}	&	\underline{01010}		&	11101				&	00010		&	11001					&	00110		&					&										 \\
				 01011 			&	10100				&	11011				&	00100		&	\underline{10011}					&	\underline{01100}		&					&											 \\
				10110			&	01001				&	\underline{10111}		&	\underline{01000}		&	00111					&	11000		&					&											 \\
				01101			&	10010				&	01111				&	10000		&	01110					&	10001		&					&									
\end{tabular}
\end{center}

\vspace{-0.1in}

\noindent
%Thus $\PCRrep{5} = \{00000, 00001, 00011, 00101, 11100, 11010, 11110, 11111\}$.
\end{exam}

The PCR is the underlying feedback function used to construct the Prefer-1 greedy construction corresponding to the lexicographically largest de Bruijn sequence.  It has also been applied in some of the simplest and most efficient de Bruijn sequence constructions~\cite{grandma2,framework,wong}.   In these constructions, the cycle representatives relate to the lexicographically smallest (or largest) strings in each cycle and they can be determined in $O(n)$ time using $O(n)$ space using standard techniques~\cite{booth,duval}.   We also apply these methods to efficiently determine the RL-reps and the RL2-reps. 

Clearly $0^n$ and $1^n$ are both RL-reps.  Consider a string $\omega = w_1w_2\cdots w_n$ in a cycle $\bP$ with RLE $r_1r_2\cdots r_\ell$ where $\ell > 1$.   
If $\omega$ is an RL-rep, then $w_1\neq w_n$ because otherwise $w_nw_1\cdots w_{n-1}$ has a larger RLE than $\omega$.  All strings in $\bP$ that differ in the first and last bits form an equivalence class under rotation with respect to their RLE.  By definition, the RL-rep will be one that is lexicographically largest  amongst all its rotations.  As noted above, such a test can be performed in $O(n)$ time using $O(n)$ space.  There is one special case to consider: when both a string beginning with 0 and its complement beginning with 1 belong to the same cycle.  For example, consider $00101101$ and $11010010$ which both have RLE 211211.  Note this RLE has period $p=3$ and it is maximal amongst its rotations.  By definition, the string beginning with 0 is not an RL-rep.  It is not difficult to see that  such a string occurs precisely  when $w_1 = 0$ and $p$ is odd, where $p$ is the period of $r_1r_2\cdots r_\ell$.

\begin{proposition}
Let  $\omega = w_1w_2\cdots w_n$ be a string with RLE $r_1r_2\cdots r_\ell$, where $\ell > 1$, in a cycle $\bP$ induced by the $\PCR$.   Let $p$ be the period of $r_1r_2\cdots r_\ell$.
Then $\omega$ is the RL-rep for $\bP$  if  and only if  
\begin{enumerate}
\item $w_1 \neq w_n$, 
\item $r_1r_2\cdots r_\ell$ is lexicographically largest amongst all its rotations,  and 
\item either $w_1 = 1$ or $p$ is even.
\end{enumerate}
Moreover, testing whether or not $\omega$ is an RL-rep can be done in $O(n)$ time using $O(n)$ space.
\end{proposition}

In a similar manner we consider RL2-reps.  Again $0^n$ and $1^n$ are both clearly RL2-reps.  Consider a string $\omega = w_1w_2\cdots w_n$ in a cycle $\bP$ with run length greater than one.
If $\omega$ is an RL2-rep, then $w_1\neq w_2$ because otherwise $w_2\cdots w_nw_1$ has a smaller RLE than $\omega$.    Thus, consider all strings $s_1s_2\cdots s_n$ in a cycle $\bP$
such that $s_2 \neq s_1$.  One of these strings is the RL2-rep.  Now consider all left rotations of these strings taking the form $s_2\cdots s_ns_1$.  Notice that a string in the latter set with the smallest RLE
will correspond to the RL2-rep after rotating the string back to the right.   As noted in the RL-case, the set of rotated strings form an equivalence class under rotation  with respect to their RLE, since their first and last bits differ.  Again, the same special case arises as with RL-reps: when both a string beginning with 0 and its complement beginning with 1 belong to the same cycle.   For example, consider the cycle containing both 10100101 and 01011010.  In each string the first two bits differ.  The set of all strings in its cycle where the first two bits differ is $\{10100101, 01001011, 10010110, 01011010, 10110100, 01101001\}$.  Rotating each string to the left we get the set  $\{01001011, 10010110, 00101101, 10110100, 01101001, 11010010\}$.  The corresponding RLEs for this latter set are $\{112112, 121121, 211211, 112112, 121121, 211211\}$.   In this case there are two strings $0100100$ and $10110100$ that both have RLE 112112. Rotating these strings back to the right we have 10100101 and 01011010 which both have the lexicographically smallest RLE of 1112111 in their cycle induced by the PCR.  By definition, the string beginning with 0 will be the RL2-rep.  Thus $\omega$ is not an RL2-rep if $w_1 = 1$, $p$ is odd, and $p< \ell$, where $p$ is the period of  the RLE $r_1r_2\cdots r_\ell$ for the string $w_2\cdots w_nw_1$.

 \begin{proposition}
Let  $\omega = w_1w_2\cdots w_n$ and let  $r_1r_2\cdots r_\ell$ be the RLE of $w_2\cdots w_nw_1$, where $\ell > 1$, in a cycle $\bP$ induced by the $\PCR$.    Let $p$ be the period of $r_1r_2\cdots r_\ell$.
Then $\omega$ is the RL2-rep for $\bP$  if  and only if  
\begin{enumerate}
\item $w_1 \neq w_2$, 
\item $r_1r_2\cdots r_\ell$ is lexicographically smallest amongst all its rotations,  and 
\item either $w_1 = 0$ or $p$ is even  or $p=\ell$.
\end{enumerate}
Moreover, testing whether or not $\omega$ is an RL2-rep can be done in $O(n)$ time using $O(n)$ space.
\end{proposition}

%Na\"{i}vely,  comparing the RLE of a string to every rotation of the RLE requires $O(n^2)$-time.  The membership testers for RL-reps and RL2-reps outlined above provide an improvement by a factor of $n$ over this na\"{i}ve approach with respect to run-time.

%=================================================================
\subsection{The complementing cycling register (CCR)}

 The \defo{complementing cycling register}, denoted $\CCR$, is the FSR with the feedback function $f(\omega) =  \overline{w_1}$, where $\overline{w_1}$ denotes the complement  $w_1$.   Thus, $\CCR(w_1w_2\cdots w_n) = w_2\cdots w_n\overline{w_1}$.
A string and its complement will belong to the same cycle induced by the $\CCR$.

\begin{exam}  \small \label{exam:CCR}
The $\CCR$ partitions $\bB(5)$ into the following four cycles $\bC_1, \bC_2, \bC_3, \bC_4$  where the top string in bold is the RL-rep for the given cycle.  The underlined string is the RL2-rep.

\vspace{-0.1in}

\begin{center}  \footnotesize
\begin{tabular}{c @{\hskip 0.4in}  c @{\hskip 0.4in}  c @{\hskip 0.4in} c}
$\bC_1$  & 
$\bC_2$  & 
$\bC_3$  & 
$\bC_4$  \\ 
\bblue{10101}		&	\bblue{11101}				&	\bblue{11001}	&\bblue{11111}		  	\\
\underline{01010} 			&	11010			&	10010		&11110				 \\
				&	10100					&	00100		&11100					 \\
				&	01000					&	\underline{01001}		&11000			 				 \\
				&	10001					&	10011		&10000			 				 \\
				&	00010					&	00110		&00000			 			 \\
				&	00101					&	01101		&00001			 			 \\
				&	\underline{01011}			&	11011		&00011			 			 \\  
				&	10111					&	10110		&00111			 				\\
				&	01110					&	01100		&\underline{01111}			  
\end{tabular}
\end{center}

\vspace{-0.1in}

\noindent
%The first string bolded in each set is a co-necklace.  The underlined strings are the \RLCCRs. Thus $\CCRrep{5} = \{11111, 11101, 11001, 10101\}$.

\end{exam}

The $\CCR$ has been applied to efficiently construct de Bruijn sequences in variety of ways~\cite{etzion1987,framework,huang}.  An especially efficient construction applies a concatenation scheme to construct a de Bruijn sequence with discrepancy, which is the maximum difference between the number of 0s and 1s in any substring,  bounded above by $2n$~\cite{gabric-concatenation, discrep}.    

As with the $\PCR$, we discuss how to efficiently determine whether or not a given string is an RL-rep or an RL2-rep for a cycle $\bC$ induced by the $\CCR$.
Consider a string $\omega = w_1w_2\cdots w_n$ in a cycle $\bC$.     
If $\omega$ is an RL-rep, then $w_1 = w_n$ because otherwise $\overline{w_n}w_1\cdots w_{n-1}$, which is also in $\bC$,  has a larger RLE than $\omega$.  All strings in $\bC$ that agree in the first and last bits form an equivalence class under rotation with respect to their RLE (that includes strings starting with both 0 and 1 for each RLE).
By definition, the RL-rep will be one that is lexicographically largest  amongst all its rotations.  As noted in the previous subsection, such a test can be performed in $O(n)$ time using $O(n)$ space.   There are no special cases to consider here since a string and its complement always belong to the same cycle.  Thus, every RL-rep must begin with 1.

\begin{proposition}
Let  $\omega = w_1w_2\cdots w_n$ be a string with RLE $r_1r_2\cdots r_\ell$ in a cycle $\bC$ induced by the $\CCR$.  
Then $\omega$ is the RL-rep for $\bC$  if  and only if  
\begin{enumerate}
\item $w_1 = w_n = 1$ and
\item $r_1r_2\cdots r_\ell$ is lexicographically largest amongst all its rotations.
\end{enumerate}
Moreover, testing whether or not $\omega$ is an RL-rep can be done in $O(n)$ time using $O(n)$ space.
\end{proposition}

In a similar manner we consider RL2-reps.   Again, consider a string $\omega = w_1w_2\cdots w_n$ in a cycle $\bC$.     
If $\omega$ is an RL2-rep, then $w_1 \neq w_2$ because otherwise $w_2\cdots w_n\overline{w_1}$ has a smaller RLE than $\omega$.  
Consider all such strings $w_2\cdots w_n\overline{w_1}$ in a cycle $\bC$ such that $w_2 \neq w_1$.  As noted in the RL-case, all such strings form an equivalence class under rotation with respect to their RLE. Clearly, such a string that has the lexicographically smallest RLE will be the RL2-rep.   There are no special cases to consider here since a string and its complement always belong to the same cycle.  Thus, every RL2-rep must begin with 0 and hence $w_2 = 1$.

\begin{proposition}
Let  $\omega = w_1w_2\cdots w_n$ be a string with RLE $r_1r_2\cdots r_\ell$ in a cycle $\bC$ induced by the $\CCR$.  
Then $\omega$ is the RL2-rep for $\bC$  if  and only if  
\begin{enumerate}
\item $w_1 = 0$ and $w_2=1$, and
\item $r_1r_2\cdots r_\ell$ is lexicographically smallest amongst all its rotations.
\end{enumerate}
Moreover, testing whether or not $\omega$ is an RL2-rep can be done in $O(n)$ time using $O(n)$ space.
\end{proposition}

%Na\"{i}vely,  comparing the RLE of a string to every rotation of the RLE requires  $O(n^2)$-time.  The membership testers for RL-reps and RL2-reps outlined above provide an improvement by a factor of $n$ over this na\"{i}ve approach with respect to run-time.

%=========================================================================================
%=================================================================
%=================================================================
\subsection{The pure run-length register ($\PRR$)} \label{sec:PRR}

The feedback function of particular focus in this paper is $f(\omega) = w_1\oplus w_2 \oplus w_n$.   
% Previous to this work, the authors are not aware of any study that applies this feedback function. 
We will demonstrate that FSR based on this feedback function partitions $\bB(n)$ into cycles of strings with the same run length.  
Because of this property, we call this FSR the \defo{pure run-length register} and denote it by $\PRR$.  
 Thus, $$\PRR(w_1w_2\cdots w_n) = w_2\cdots w_n(w_1\oplus w_2 \oplus w_n).$$
This follows the naming of the pure cycling register (PCR) and 
the pure summing register (PSR), which is based on the feedback function $f(\omega) = w_1\oplus w_2 \oplus \cdots \oplus w_n$~\cite{golomb}.

Let  $\bR_1, \bR_2, \ldots, \bR_t$ denote the cycles induced by the $\PRR$ on $\bB(n)$.  The following example illustrates how the cycles induced by the  $\PRR$ relate to the cycles induced
by the $\PCR$ and $\CCR$.

\begin{exam} \small \label{exam:PRR}
The $\PRR$ partitions $\bB(6)$ into the following 12 cycles $\bR_1, \bR_2, \ldots, \bR_{12}$ where the top string in bold is the RL-rep for the given cycle. The underlined string is the RL2-rep.  The cycles are ordered in non-increasing order with respect to the run lengths of their RL-reps.

\vspace{-0.1in}

 \begin{center}   \footnotesize
\begin{tabular}{c @{\hskip 0.16in}  c @{\hskip 0.16in}  c @{\hskip 0.16in} c @{\hskip 0.16in} c @{\hskip 0.16in} c @{\hskip 0.16in} c @{\hskip 0.16in} c  @{\hskip 0.16in} c @{\hskip 0.16in} c   }
$\bR_1$  & 
$\bR_2$  & 
$\bR_3$  & 
$\bR_4$ & 
$\bR_5$  & 
$\bR_6$ & 
$\bR_7$ & 
$\bR_8$ & 
$\bR_9$ & 
$\bR_{10}$  \\  
	\bblue{101010}		& 	\bblue{110101}	&	\bblue{001010}		&	\bblue{111010}		&	\bblue{110010} 		 &	\bblue{111101}	& 	\bblue{000010}	& \bblue{111001}  &	\bblue{000110} &	\bblue{111110}				 \\
	\underline{010101}  			&	\underline{101011}		&	\underline{010100}			&	110100			&	100100 				&	111011		&	000100		& 110011		&	001100	&	111100					\\
					&  	010110		&	101001			&	101000			&	001001				&	110111		&	001000		& \underline{100111}		&	\underline{011000}	&	111000				 \\
					& 	101101		&	010010			&	010001			&	\underline{010011}				&	\underline{101111}		&	\underline{010000}		& 001110		&	110001	&	110000								 \\
					& 	011010		&	100101			&	100010			&	100110 				&	011110		&	100001		& 011100		&	100011	&	100000								  \\
					&				&					&	000101			&	001101				&				&	   			&			&			&	000001			 									 \\
					&				&					&	001011			&	011011				&	   			&				&			&			&	000011			 									 \\
					&				&					&	\underline{010111}			&	110110 				&	   			&				&			&			&	000111			 									\\
					&				&					&	101110			&	101100 				&	   			&				&			&			&	001111			 								 \\	
					& 				&					&	011101			&	011001 				&				&				&			&			&	\underline{011111}	 \\
$\bR_{11}$  & 
$\bR_{12}$  \\
\underline{\bblue{111111}} & \underline{\bblue{000000}} \\
			
\end{tabular}
\end{center}

\noindent
By omitting the last bit of each string, the columns are precisely the cycles of the $\PCR$ and $\CCR$ for $n=5$.  The cycles $\bR_1, \bR_4, \bR_5, \bR_{10}$ relating to the $\CCR$ start and end with the different bits.  
The remaining cycles relate to the $\PCR$; each string in these cycles start and end with the same bit. 
%The prefixes of length $5$  highlighted in the first row of each column (cycle) are all the necklaces and co-necklaces for $n=5$.   The underlined strings are the RL-reps. Thus $\RLrep{6} = \{000000, 000010, 000110,  001010, 111001, 110101, 111101, 111111\} \cup \{111110, 111010, 110010, 101010\}$.
\end{exam}

\noindent
In the example above,  note that all the strings in a given cycle $\bR_i$ have the same run length.
\begin{lemma} \label{lem:rle}
All the strings in a given cycle $\bR_i$ have the same run length.
\end{lemma}
\begin{proof}
Consider a string $\omega = w_1w_2\cdots w_n$ and the feedback function $f(\omega) = w_1\oplus w_2 \oplus w_n$.  It suffices to show that $w_2\cdots w_n f(\omega)$ has the same run length as $\omega$.
This is easily observed since if $w_1 = w_2$ then $w_n =  f(\omega)$ and if $w_1 \neq w_2$ then $w_n \neq  f(\omega)$.
\end{proof}

Based on this lemma, if the strings in $\bR_i$ have run length $\ell$, we say that $\bR_i$ has run length $\ell$.
Each cycle $\bR_i$ has another interesting property:  either all the strings start and end with the same bit, or all the strings start and end with different bits.  
If the strings start and end with the same bit, then $\bR_i$ must have odd run length and if we remove the last bit of each string we obtain a cycle induced by the $\PCR$ of order $n{-}1$.  In this case we say that $\bR_i$ is a \defo{PCR-related cycle}.  Such a cycle is \defo{periodic} if for each string 
$\omega = w_1w_2\cdots w_n \in \bR_i$, $w_1w_2\cdots w_{n-1}$ is periodic; otherwise, $\bR_i$ is \defo{aperiodic} and the cycle contains $n{-}1$ distinct strings.
 If the strings start and end with the different bits, then $\bR_i$ must have even run length and if we remove the last bit of each string we obtain a cycle induced by the $\CCR$ of order $n{-}1$.  In this case we say that $\bR_i$ is a \defo{CCR-related cycle}.   Such a cycle is \defo{periodic} if for each string 
$\omega = w_1w_2\cdots w_n \in \bR_i$, $w_1w_2\cdots w_{n-1}\overline{w_1w_2\cdots w_{n-1}}$ is periodic; otherwise, it is \defo{aperiodic} and the cycle contains $2n-2$ distinct strings.  As an example, consider the CCR-related cycle for $n=7$ containing the strings $\{001100\red{1}, 011001\red{0}, 110011\red{0}, 100110\red{1}\}$.  Consider $\omega = 001100\red{1}$ and note that $001100\underline{110011}$ is periodic.
These observations were first made in~\cite{sala} and are illustrated in Example~\ref{exam:PRR}, where the periodic cycles are $\bR_1$, $\bR_{11}$ and $\bR_{12}$.   

%
%=================
% OBSERVATION PRR
%=================
\begin{observation} \label{obs:PRR}
Let $\omega, \omega' \in \bR_i$ and let $\PRR^j(\omega) = \omega'$.  
If $\bR_i$  is PCR-related then $\PRR^{(n-1)-j}(\omega') = \omega$.
If $\bR_i$  is CCR-related then  $\PRR^{(2n-2)-j}(\omega') = \omega$ and furthermore $\PRR^{n-1}(\omega) = \overline{\omega}$.
\end{observation}

The following lemma considers the RLEs for strings in a cycle $\bR_i$.

\begin{lemma} \label{lem:RLEs}
Let $\omega = w_1w_2\cdots w_n$ be a string in $\bR_i$ with RLE of the form  $1r_1r_2\cdots r_m$ or
$r_1r_2\cdots r_m1$.  Then the RLE of any string in $\bR_i$ has the form
$$(r_s{-}j)r_{s+1}\cdots r_mr_1\cdots r_{s-1} (j{+}1),$$ 
for some $1 \leq s \leq m$ and $0 \leq j < r_s$.
\end{lemma}
\begin{proof}
If the RLE of $\omega$ begins with 1 then $w_1 \neq w_2$ and thus $\PRR(\omega) = w_2\cdots w_n\overline{w_n}$ will have
RLE of the form $r_1r_2\cdots r_m1$.
Starting with this RLE, the next $r_1-1$ applications of the $\PRR$ yield strings with RLE: 
$$ \blue{r_1r_2\cdots r_m1}, \ \ \ (r_1{-}1)r_2\cdots r_m2, \ \  \  (r_1{-}2)r_2\cdots r_m3, \ \  \ \ldots, \ \  1r_2\cdots r_mr_1.$$  
Repeating this pattern produces the remaining strings in $\bR_i$, which leads to the desired result.
\end{proof}

\begin{exam}
Consider RL-rep $\omega = 0001110110$ belonging to an aperiodic PCR-related cycle  containing the following nine strings with their RLE in parentheses:
\begin{center}
\begin{tabular}{l}
0001110110 (\blue{33121}), \ \   0011101100 (\blue{23122}), \ \   0111011000 (\blue{13123}), \\
1110110001 (\blue{31231}), \ \    1101100011 (\blue{21232}), \ \     1011000111 (\blue{11233}), \\
0110001110 (\blue{12331}), \\  
1100011101 (\blue{23311}), \ \     1000111011 (\blue{13312}).  
\end{tabular}
\end{center}
 \vspace{-0.15in}
 \end{exam}

%The following Corollary follows from the previous lemma
%and the definition of an RL-rep, noting that the RLE of the RL-rep of $\bR_i$ ends with 1.   
%
%\begin{corollary}~\label{cor:RL-rep}
%Let $\omega, \sigma \in \bR_i$, where $\sigma$ is the RL-rep.
%If $\sigma$ has RLE $R_i = r_1r_2\cdots r_m1$ and   
% $\omega$ has RLE $s_1s_2\cdots s_{m'}1$ 
%where $s_j \cdots s_{m'}$ is a prefix of $R_i$ for some $1<j \leq m'$, then $r_1\cdots r_m$ is periodic and 
%$\omega = \sigma$ or $\omega = \overline{\sigma}$.
%\end{corollary}

We can apply the RL-rep and RL2-rep testers for cycles induced by the $\PCR$ and $\CCR$ to determine whether or not a string $\omega$ is an RL-rep or an RL2-rep for a cycle $\bR_i$.  These testers, outlined in the following propositions, are critical to the efficiency of our upcoming de Bruijn successors.    

\begin{proposition} \label{fact:RL-test}
Let  $\omega = w_1w_2\cdots w_n$ be a string in a cycle $\bR_i$.  
Then $\omega$ is the RL-rep for $\bR_i$  if  and only if  
\begin{enumerate}
\item $w_1 = w_n$ and $w_1w_2\cdots w_{n-1}$ is an RL-rep with respect to the $\PCR$, or 
\item $w_1 \neq w_n$ and $w_1w_2\cdots w_{n-1}$ is an RL-rep with respect to the $\CCR$. 
\end{enumerate}
Moreover, testing whether or not $\omega$ is an RL-rep for $\bR_i$ can be done in $O(n)$ time using $O(n)$ space.
\end{proposition}
\begin{proposition}  \label{fact:RL2-test}
Let  $\omega = w_1w_2\cdots w_n$ be a string in a cycle $\bR_i$.  
Then $\omega$ is the RL2-rep for $\bR_i$  if  and only if  
\begin{enumerate}
\item $w_1 = w_n$ and $w_1w_2\cdots w_{n-1}$ is an RL2-rep with respect to the $\PCR$, or 
\item $w_1 \neq w_n$ and $w_1w_2\cdots w_{n-1}$ is an RL2-rep with respect to the $\CCR$. 
\end{enumerate}
Moreover, testing whether or not $\omega$ is an RL2-rep for $\bR_i$ can be done in $O(n)$ time using $O(n)$ space.
\end{proposition}
The above propositions can easily be verified by the reader based on the definitions of RL-reps and RL2-reps and applying Lemma~\ref{lem:RLEs}.

%\begin{corollary}~\label{cor:RLEs}
%Let $\omega = w_1w_2\cdots w_{n-1}1$ be a string in $\bR_i$ with run length less than $n$.
%If the RLE of $\omega$  has a proper suffix equal to a prefix, then $\omega$ is not an RL-rep.
%%If the RLE of $\omega$ begins with 1 and has a proper prefix equal to a suffix, then $\omega$ is not an RL2-rep.
%\end{corollary}
%For example, consider $\omega \in \bR_i$ with RLE 411321\underline{411}.  It is not an RL-rep since by Lemma~\ref{lem:RLEs}, a string with RLE
%414113211 is also in $\bR_i$.  

%=================================================================
\section{Generic de Bruijn successors based on the PRR} \label{sec:generic}

%In this section we apply Theorem 3.5  from~\cite{framework} to the cycles $\bR_1, \bR_2, \ldots, \bR_t$ induced by the PRR on $\bB(n)$.  We then 
In this section we provide two generic de Bruijn successors that are applied  to derive specific de Bruijn successors for $\same{n}$ and $\opposite{n}$ in the subsequent sections.
The results relate specifically to the PRR and we assume that $\bR_1, \bR_2, \ldots, \bR_t$ denote the cycles induced by the $\PRR$ on $\bB(n)$.

Let $\omega = w_1w_2 \cdots w_n$ be a binary string.  Define the \defo{conjugate} of  $\omega$ to be $\hat \omega = \overline{w_1}w_2\cdots w_n$.  Similar to Hierholzer's cycle-joining approach discussed in Section~\ref{sec:euler},  Theorem 3.5 from~\cite{framework} can be applied to systematically join together the \emph{ordered} cycles $\bR_1, \bR_2, \ldots, \bR_t$ given certain representatives $\alpha_i$ for each $\bR_i$.  This theorem is restated as follows when applied to the PRR and the function $f(\omega) =  w_1 \oplus w_2 \oplus w_n$.

\begin{theorem} \label{thm:framework}
For each $1< i \leq t$, if the conjugate $\hat \alpha_i$ of  the representative $\alpha_i$ for cycle $\bR_i$ belongs to some $\bR_j$ where $j < i$, then 
\begin{center}
$g(\omega) = \left\{ \begin{array}{ll}
        \overline{f(\omega)} &\ \  \mbox{if $\omega$ or $\hat \omega$ is in $\{\alpha_2, \alpha_3, \ldots ,\alpha_t\}$;}\\
        f(\omega) &\ \  \mbox{otherwise}\end{array} \right.$
\end{center}
is a de Bruijn successor.
\end{theorem}
Together, the ordering of the cycles and the sequence  $\alpha_2, \alpha_3, \ldots , \alpha_t$ correspond to a rooted tree, where the nodes are the cycles $\bR_1, \bR_2, \ldots, \bR_t$ with $\bR_1$ designated as the root.  There is an edge between two nodes $\bR_i$ and $\bR_j$ where $i>j$, if and only if $\hat \alpha_i$ is in $\bR_j$;  we say that $\bR_j$ is the \defo{parent} of $\bR_i$.  Each edge represents the joining of two cycles similar to the technique used in Hierholzer's Euler cycle algorithm (see Section~\ref{sec:euler}).   An example of such a tree for $n=6$ is given in the following example.

\begin{exam}  \small  \label{exam:tree}
Consider the cycles $\bR_1, \bR_2, \ldots , \bR_{12}$ for $n=6$  from Example~\ref{exam:PRR} along with their corresponding RL-reps $\alpha_i$ for each $\bR_i$.  For each $i>1$, $\hat \alpha_i$ belongs to some $\bR_j$ where $j<i$. Thus, we can apply Theorem~\ref{thm:framework} to obtain a de Bruijn successor $g(\omega)$ based on these representatives.  The following tree illustrates the joining of these cycles based on $g$: 
\begin{center}
\resizebox{3.3in}{!}{\includegraphics{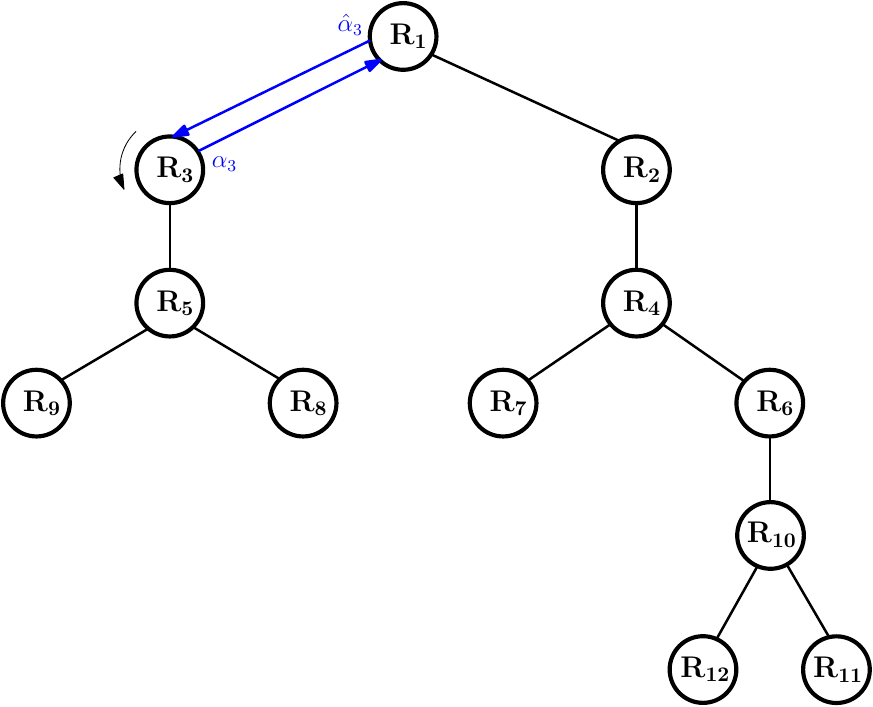}}
\end{center}
Starting with 101010 from $\bR_1$, and repeatedly applying the function $g(\omega)$ we obtain the de Bruijn sequence:
$$101010010011\underline{000110}1100111001010110100010000101110111100000011111.$$
%This sequence corresponds to {\sc DB}($g$,011111).
Note that the RL-rep of $\bR_3$ is $\alpha_3=001010$  and its conjugate $\hat \alpha_3=101010$  is found in its parent $\bR_1$.  
The last string visited in each cycle $\bR_i$, for $i > 1$, is its representative $\alpha_i$.  
%The order these representatives appear correspond to a post-order traversal of the tree; $\alpha_9=0001110$ is the first RL-rep to appear in the sequence.
%
\end{exam}
%As an example, the 12 cycles for $n=6$ based on the cycle listing from Example~\ref{exam:PRR} are joined as illustrated in Figure~\ref{fig:tree}, where each $\alpha_i$ is the RL-rep for $\bR_i$.

%---------------------
%\begin{figure}[h]
%\begin{center}
%\resizebox{3.3in}{!}{\includegraphics{"../RL-tree"}}
%\end{center}

%\vspace{-0.15in}
%\captionsetup{format=hang}
%\caption{  \small
%Illustrating the joining of the cycles $\bR_1, \bR_2, \ldots , \bR_{12}$ listed in Example~\ref{exam:PRR}, where each $\alpha_i$ is the RL-rep for $\bR_i$.  Note that
%the RL-rep of $\bR_3$ is $\alpha_3=001010$  and its conjugate $\hat \alpha_3=101010$  is found in $\bR_1$.  Starting from the string 101010 from $\bR_1$, and applying
%the function $g(\omega)$ from Theorem~\ref{thm:framework} based on the RL-reps we obtain the de Bruijn sequence 
%$$101010010011\underbrace{000110}_{\alpha_9}1100111001010110100010000101110111100000011111$$
%The last string visited in each cycle $\bR_i$, for $i > 1$, is its representative.  The order these representatives appear correspond to a post-order traversal of the tree; $\alpha_9$ appears first in the sequence.
%}
%\label{fig:tree}
%\end{figure}
%---------------------

\noindent
The following observations, which will be applied later in our more technical proofs,  follow from the tree interpretation of the ordered cycles rooted at $\bR_1$ from Theorem~\ref{thm:framework} as illustrated in the previous example.\footnote{Note that Theorem~\ref{thm:framework} and Observation~\ref{obs:tree} apply more generally to any non-singular feedback function $f$.}

%+====================
\begin{observation} \label{obs:tree}
Let $g$ be a de Bruijn successor from Theorem~\ref{thm:framework} based on representatives $\alpha_2, \alpha_3, \ldots, \alpha_t$. 
Let $\mathcal{D}_n = DB(g, w_1w_2\cdots w_n)$ and let $\mathcal{D'}_n = w_2\cdots w_n\mathcal{D}_n$ denote a linearized de Bruijn sequence.   
If the length $n$ prefix of $\mathcal{D'}_n$ is in $\bR_1$,  then for each $1< i  \leq t$: 

\begin{enumerate}
\item $\hat \alpha_i$ appears before all strings in $\bR_i$, 
\item the $m$ strings of $\bR_i$ appear in the following order: $\PRR(\alpha_i), \PRR^2(\alpha_i), \ldots , \PRR^m(\alpha_i) = \alpha_i$,
\item if $\bR_i$ and $\bR_k$ are on the same level in the corresponding tree of cycles rooted at $\bR_1$, then either every string in $\bR_i$ comes before every string in $\bR_k$ or vice-versa,
\item the strings in all descendant cycles of $\bR_i$ appear after $\hat \alpha_i$ and before $\alpha_i$, and
\item if $\hat \alpha_i = a_1a_2\cdots a_n$, then $a_2\cdots a_n g(\hat \alpha_i)$ is in $\bR_i$. 

\end{enumerate}
\end{observation}
%=============

As an application of Theorem~\ref{thm:framework}, consider the cycles $\bR_1, \bR_2, \ldots, \bR_t$  to be ordered in \emph{non-increasing order} based on the run length of each cycle.  
Such an ordering is given in Example~\ref{exam:PRR} for $n=6$.  Using this ordering, let $\alpha_i = a_1a_2\cdots a_n$  be any string in $\bR_i$, for $i > 1$,  such that $a_1 = a_2$. Note that $\hat \alpha_i$  has run length that is one \emph{more} than the run length of $\alpha_i$ and thus  $\hat \alpha_i$  belongs to some $\bR_j$ where $j < i$.   
Thus, Theorem~\ref{thm:framework}  can be applied to describe the following generic de Bruijn successor based on the $\PRR$.

%
%=============
\begin{result}
\vspace{-0.2in}
\begin{theorem}  \label{thm:RL}
Let $\bR_1, \bR_2, \ldots, \bR_t$ be listed in non-increasing order with respect to the run length of each cycle.  Let $\alpha_i = a_1a_2\cdots a_n$ denote a representative in $\bR_i$ such that $a_1=a_2$, for each $1 < i \leq t$.   Let $\omega = w_1w_2\cdots w_n$ and let $f(\omega) =  w_1 \oplus w_2 \oplus w_n$.  Then the function:
\begin{center}
$g(\omega) = \left\{ \begin{array}{ll}
        \overline{f(\omega)} &\ \  \mbox{if $\omega$ or $\hat \omega$ is in $\{\alpha_2, \alpha_3, \ldots ,\alpha_t\}$;}\\
        f(\omega) &\ \  \mbox{otherwise.}\end{array} \right.$
\end{center}
is a de Bruijn successor.

\end{theorem}
\vspace{-0.2in}
\end{result}

\noindent
%Since the cycles containing strings with the same run length $r$ all appear at the same level in the related tree of cycles rooted at $\bR_1$ (see for instance Example~\ref{exam:tree}), we make the following observation.

%\begin{observation} \label{obs:treeS}
%Consider the tree of cycles  with nodes $\bR_1, \bR_2, \ldots , \bR_t$  induced by the ordering and representatives from Theorem~\ref{thm:RL}, and the corresponding function $g$.
%Let $\omega = w_1w_2\cdots w_n$ be a string in the root $\bR_1$ and let $\bR_i$ and $\bR_j$ have strings with the same run length, where $i,j > 1$. 
%Then either every string in $\bR_i$ comes before every string in $\bR_j$ in $w_2\cdots w_n${\sc DB}$(g,\omega)$   or vice-versa.
%\end{observation}

Now consider the cycles $\bR_1, \bR_2, \ldots, \bR_t$  to be ordered in \emph{non-decreasing order} based on the run length of each cycle.   This means the first two cycles $\bR_1$ and $\bR_2$ will be the cycles containing $0^n$ and $1^n$.   But given this ordering, there is no way to satisfy Theorem~\ref{thm:framework} since the conjugate of any representative for $\bR_2$ will not be found in $\bR_1$.  However, if we let $\bR_t = \{1^n \}$, and order the remaining cycles in \emph{non-decreasing order} based on the run length of  each cycle, then we obtain a result similar to Theorem~\ref{thm:RL}.  
Observe, that this relates to the special case described for the Prefer-opposite greedy construction illustrated in Figure~\ref{fig:trees}.
Using this ordering, let $\alpha_i = a_1a_2\cdots a_n$  be any string in $\bR_i$, for $1<i<t$,  such that $a_1 \neq a_2$. Such a string exists since $\bR_1 = \{0^n\}$ and $\bR_{t} = \{1^n\}$.  This means $\hat \alpha_i$ has run length that is one \emph{less} than the run length of $\alpha_i$ and thus $\hat \alpha_i$ belongs to some $\bR_j$ where $j < i$.  For the special case when $i=t$, the conjugate of $1^n$ clearly is found in some $\bR_j$ where $j<t$.    Thus,  Theorem~\ref{thm:framework} can be applied again to describe another generic de Bruijn successor based on the $\PRR$.

%=============
\begin{result}
\vspace{-0.2in}
\begin{theorem}  \label{thm:RL2}
Let $\bR_t = \{1^n\}$ and let the remaining cycles $\bR_1, \bR_2, \ldots, \bR_{t-1}$ be listed in non-decreasing order with respect to the run length of each cycle.  Let $\alpha_i = a_1a_2\cdots a_n$ denote a representative in $\bR_i$ such that $a_1\neq a_2$, for each $1 < i < t$.   Let $\omega = w_1w_2\cdots w_n$ and let $f(\omega) =  w_1 \oplus w_2 \oplus w_n$.  Then the function:
\begin{center}
$g_2(\omega) = \left\{ \begin{array}{ll}
        \overline{f(\omega)} &\ \  \mbox{if $\omega$ or $\hat \omega$ is in $\{\alpha_2, \alpha_3, \ldots ,\alpha_t\}$;}\\
        f(\omega) &\ \  \mbox{otherwise.}\end{array} \right.$
\end{center}
is a de Bruijn successor.

\end{theorem}
\vspace{-0.2in}
\end{result}
%=============
%
\noindent
%Again, the cycles containing strings with the same run length $r$ all appear at the same level in the related tree of cycles, resulting in the following observation.

%\begin{observation}  \label{obs:treeO}
%Consider the tree of cycles  with nodes $\bR_1, \bR_2, \ldots , \bR_t$  induced by the ordering and representatives from Theorem~\ref{thm:RL2}, and the corresponding function $g_2$.
%Let $\omega = w_1w_2\cdots w_n$ be a string in the root $\bR_1$ and let $\bR_i$ and $\bR_j$ have strings with the same run length, where $i,j > 1$. 
%Then either every string in $\bR_i$ comes before every string in $\bR_j$ in $w_2\cdots w_n${\sc DB}$(g_2,\omega)$   or vice-versa.
%\end{observation}

When Theorem~\ref{thm:RL} and Theorem~\ref{thm:RL2}  are applied na\"{i}vely, the resulting de Bruijn successors are not efficient since storing the set $\{\alpha_2, \alpha_3, \ldots ,\alpha_t\}$ requires exponential space.  However, if a membership tester for the set can be defined efficiently, then there is no need for the set to be stored.  Such sets of representatives are presented in the next two sections.

\begin{comment}

As an example, if each representative $\alpha_i$ of $\bR_i$ corresponds to its  RL-rep, then the representative will begin with 00 or 11 for all $i > 1$.  
Similarly,  if each representative $\alpha_i$ of $\bR_i$ corresponds to its  RL2-rep, then the representative will begin with 01 or 10 for all $i > 1$.  
%Thus, the run length of the conjugate $\hat \alpha_i$ will be one more than the run length of $\alpha_i$.  
This leads directly to the following corollary based on the above two theorems.  

%
\begin{result}
\vspace{-0.2in}
\begin{corollary}
Let $\omega = w_1w_2\cdots w_n$ and let $f(\omega) =  w_1 \oplus w_2 \oplus w_n$. 
The functions
%
\begin{center}
$RL(\omega) = \left\{ \begin{array}{ll}
        \overline{f(\omega)} &\ \  \mbox{ if $\omega$ or $\hat \omega$ is in  $\RLrep{n}$;}\\
        f(\omega) &\ \  \mbox{otherwise,}\end{array} \right.$
\end{center}
%
and
%
\begin{center}
$RL2(\omega) = \left\{ \begin{array}{ll}
        \overline{f(\omega)} &\ \  \mbox{ if $\omega$ or $\hat \omega$ is in  $\RLLrep{n}$;}\\
        f(\omega) &\ \  \mbox{otherwise,}\end{array} \right.$
\end{center}
%
%
are both de Bruijn successors.
\end{corollary}
%
\vspace{-0.2in}
\end{result}
%=============

\end{comment}

%=========================================================================================
%=================================================================
%=================================================================
\section{A de Bruijn successor for $\same{n}$}  \label{sec:same}

%In this section we present a de Bruijn successor that can be used to construct $\same{n}$ in $O(n)$ time per bit using $O(n)$ space.  We begin by presenting an intermediate de Bruijn successor that is used as a building block in our successor-rule for $\same{n}$.   

In this section we define a de Bruijn successor for $\same{n}$.    Recall the partition $\bR_1, \bR_2, \ldots , \bR_t$ of $\bB(n)$ induced by the $\PRR$.  In addition to the RL-rep, we define a new representative for each cycle, called the LC-rep,  where the LC stands for Lexicographic Compositions which are further discussed in Section~\ref{sec:LC}.  Then,  considering these two representatives along with a small set of special strings, we define a third representative, called the same-rep.  For each representative, we can apply Theorem~\ref{thm:RL}  to produce a new de Bruijn successor.    The definitions for these three representatives are as follows:
\begin{itemize}
\item \defo{RL-rep}: The string with the lexicographically largest RLE; if there are two such strings, it is the one beginning with 1.
\item \defo{LC-rep}:  %The string $\omega$ such that $r{+}1$ applications of the PRR starting with $\omega$ yields the RL-rep, where $r$ is the number of consecutive 1s at the end of the RLE of the RL-rep.  
The RL-rep for cycles with run length 1 and $n$.
%The strings $0^n$ and $1^n$ for the classes $\{0^n\}$ and $\{1^n\}$ respectively.  
For all other classes, it is the string $\omega$ with RLE $21^{i-1}r_{i+1}\cdots r_{\ell}$  where $i=\ell$ or $r_{i+1} \neq 1$ such that
$\PRR^{i+1}(\omega)$ is  the RL-rep.  %Alternatively, it is the string $\omega$ such that $r{+}1$ applications of the PRR starting with $\omega$ yields the RL-rep, where $r$ is the number of consecutive 1s at the end of the RLE of the RL-rep.
\item \defo{same-rep}: 
 $\left\{ \begin{array}{ll}
        \mbox{RL-rep} &\ \  \mbox{if the RL-rep is \emph{same-special}}\\
        \mbox{LC-rep}\  &\ \  \mbox{otherwise.}\end{array} \right.$        
\end{itemize}      
We say an RL-rep is \defo{same-special} if it belongs to the set $\Special(n)$ defined as follows:

\begin{quote}
 $\Special(n)$ is the set of length $n$ binary strings that begin and end with 0 and have RLE of the form $(21^{2x})^y1^z$, where $x \geq 0$, $y \geq 2$, and $z \geq 2$.
 \end{quote}
The RL-reps have already been illustrated in Section~\ref{sec:feedback}.  There are relatively few strings in $\Special(n)$ and they all have odd run length since they begin and end with 0; they belong to PCR-related cycles.  The need for identifying same-special strings is revealed in the proof for the upcoming Proposition~\ref{fact:tough}.

\begin{exam}  \small  \label{exam:same-rep}  
The RLE of the strings in $\Special(n)$ for $n=10,11,12,13$.
\medskip

\noindent
$n=10$:  \ \ 2221111  \\
$n=11$: \ \  2222111,  221111111, 211211111 \\
$n=12$: \ \ 2222211, 222111111 \\
$n=13$: \ \ 222211111, 22111111111, 21121111111 
%$\Samerep{10}$  \   =  \  $\LCrep{10}  \ \setminus \  \{2111221\} \ \cup \  \{2221111\}$  \\
%$\Samerep{11}$   \ =   \ $\LCrep{11}  \ \setminus \  \{2112221, 211111121, 211112111\} \ \cup \  \{2222111,  221111111, 211211111 \}$  \\
%$\Samerep{12}$  \  =  \ $\LCrep{12}  \ \setminus \  \{2122221, 211111221\} \ \cup \  \{22222211, 222111111\}$  \\
%$\Samerep{13}$  \  =  \ $\LCrep{13}  \ \setminus \  \{211112221, 21111111121, 21111112111\} \ \cup \  \{222211111, 22111111111, 21121111111 \}$  
\end{exam}

To illustrate an LC-rep, consider the string $\omega = \underline{11010}1111011$ with RLE $2111412$.  The string $\omega$ is an LC-rep since $PRR^{5}(\omega) = 111101110101$ which
is an RL-rep with RLE $413\underline{1111}$.  Note that another way to define the LC-rep is as follows: If the RLE of an RL-rep ends with $i$ consecutive 1s, then the corresponding LC-rep is the string $\omega$ such that $PRR^{i+1}(\omega)$ is the RL-rep.

% To illustrate an LC-rep, consider the RL-rep $111101110101$ with RLE $413\underline{1111}$.  There are $r=4$ consecutive 1s in the suffix of the RLE.  Applying the PRR $r+1=5$ times to $111101110101$ we obtain the LC-rep  $\underline{11010}1111011$.   Note that RLE of an LC-rep will always start with 2 except for the case
%when the corresponding RL-rep has run length $n$.

Let $\RLrep{n}$,  $\LCrep{n}$, and $\Samerep{n}$  denote the sets  of all  length $n$ RL-reps, LC-reps, and same-reps, respectively, {\bf not including} the representative with run length $n$.   Consider the following feedback functions where $\omega = w_1w_2\cdots w_n$ and $f(\omega) =  w_1 \oplus w_2 \oplus w_n$:

\begin{center}
$RL(\omega) = \left\{ \begin{array}{ll}
            \overline{f(\omega)} &\ \  \mbox{if $\omega$ or $\hat \omega$  is in $\RLrep{n}$;}\\
        f(\omega) &\ \  \mbox{otherwise,}\end{array} \right.$

\bigskip

$LC(\omega) = \left\{ \begin{array}{ll}
            \overline{f(\omega)} &\ \  \mbox{if $\omega$ or $\hat \omega$  is in $\LCrep{n}$;}\\
        f(\omega) &\ \  \mbox{otherwise,}\end{array} \right.$
        
        \bigskip

$S(\omega) = \left\{ \begin{array}{ll}
            \overline{f(\omega)} &\ \  \mbox{if $\omega$ or $\hat \omega$  is in $\Samerep{n}$;}\\
       	    f(\omega) &\ \  \mbox{otherwise.}\end{array} \right.$

\end{center}

\begin{theorem} \label{thm:successorS}
The feedback functions $RL(\omega)$, $LC(\omega)$ and $S(\omega)$ are de Bruijn successors.
\end{theorem}
\begin{proof}
Let the partition $\bR_1, \bR_2, \ldots, \bR_t$ of $\bB(n)$ induced by the PRR be  listed in non-increasing order with respect to the run length of each cycle.  
Observe that $\bR_1$ is the cycle whose strings have run length $n$, and thus any representative of $\bR_1$ will have run length $n$.  
By definition, this representative is not in the sets $\RLrep{n}$,  $\LCrep{n}$, and $\Samerep{n}$.   Now consider $\bR_i$ for $i > 1$.
Clearly the RL-rep for  $\bR_i$ will begin with 00 or 11 and by definition,  the LC-rep for $\bR_i$  also begins with 00 or 11.      
Together these results imply that each same-rep for $\bR_i$  will also begin with 00 or 11.
Thus, it follows directly from Theorem~\ref{thm:RL} that $RL(\omega)$, $LC(\omega)$ and $S(\omega)$ are de Bruijn successors.
\end{proof}

Recall that $alt(n)$ denotes the alternating sequence of 0s and 1s of length $n$ that ends with 0. 
Let $\mathcal{X}_n = x_1x_2 \cdots x_{2^n}$  be the de Bruijn sequence returned by {\sc DB}($S, 0 alt(n{-}1)$); it will have suffix equal to the seed $0 alt(n{-}1)$. 
Let  $\mathcal{X}'_n$ denote the linearized de Bruijn sequence $alt(n{-}1) \mathcal{X}_n$.
Our goal is to show that  $\mathcal{X}_n = \same{n}$.   Our proof applies the following two propositions.

\begin{proposition} \label{fact:prefixS}
$\mathcal{X}_n$ has prefix $1^n$.
\end{proposition}
\begin{proof}
The result follows from $n$ applications of the successor $S$ to the seed $0 alt(n{-}1)$.
\end{proof}

%First, we state some preliminary facts about the sequence  $\mathcal{X}_n$ and same-reps. 
%For each fact below, assume the partition $\bR_1, \bR_2, \ldots, \bR_t$ of $\bB(n)$ induced by the PRR be  listed in non-increasing order with respect to the run length of the strings in each cycle.  
%The first  fact is trivial to observe.   

%\begin{proposition} \label{fact:conjugate}
%If $\alpha_i$ is a same-rep for $\bR_i$ where $i > 1$, then the run length of $\alpha_i$ is one less than the run length of $\hat \alpha_i$.
%\end{proposition}

%\noindent
%Since the seed $0 alt(n{-}1)$ is used to generate $\mathcal{Y}_n$,  the suffix must also be $0 alt(n{-}1)$.  Applying $n$ applications of the successor $S(\omega)$ to this string we obtain the following fact. 

%\begin{proposition} \label{fact:Xpre}
%$\mathcal{X}_n$ has prefix $1^n$ and has suffix $alt(n{-}1)$.
%\end{proposition}

%\noindent
%A proof for the next fact is given later in Section~\ref{sec:proofs}.  
%It relies on a deeper study of the RL-reps of the conjugates for each string in a given cycle $\bR_i$, and
%a more specific ordering of the cycles $\bR_1, \bR_2, \ldots, \bR_t$.

\begin{proposition} \label{fact:tough}
%If  the run length of $\omega  = w_1w_2\cdots w_n$ is  one more than the run length of $\hat \omega$ and neither $\omega$ nor $\hat \omega$ are same-reps,  then $\hat \omega$ comes before $\omega$ in $\mathcal{X}'_n$.
If $\beta$ is a string in $\bB(n)$ such that the run length of $\beta$ is one more than the run length of $\hat \beta$ and neither $\beta$ nor $\hat \beta$ are same-reps,  then $\hat \beta$ appears before $\beta$ in $\mathcal{X}'_n$.
\end{proposition}

\noindent A proof of this proposition is given later in Section~\ref{sec:proof1}.

%------------------
\begin{result}
\vspace{-0.2in}
\begin{theorem} \label{thm:main}
The de Bruijn sequences $\same{n}$ and $\mathcal{X}_n$ are the same. 
\end{theorem}
\vspace{-0.2in}

\end{result}
\noindent
\begin{proof}
Let $\same{n}  = s_1s_2\cdots s_{2^n}$, let $\mathcal{X}_n = x_1x_2\cdots x_{2^n}$.   Recall that $\mathcal{X}_n$ ends with $alt(n{-}1)$.  
From Proposition~\ref{fact:Spre} and Proposition~\ref{fact:prefixS}, $x_1x_2\cdots x_n = s_1s_2\cdots s_n = 1^n $ and  moreover $\same{n}$ and $\mathcal{X}_n$ share the same length $n{-}1$ suffix.
Suppose there exists some smallest $t$, where $n < t \leq 2^n$, such that $s_t \neq x_t$.   Let $\beta = x_{t-n}\cdots x_{t-1}$ denote the length $n$ substring of $\mathcal{X}_n$ ending at position $t{-}1$.
Then $x_t \neq x_{t-1}$, because otherwise the RLE of $\mathcal{X}_n$ is lexicographically larger than that of $\same{n}$, contradicting Proposition~\ref{fact:rle}.
We claim that $\hat \beta$ comes before $\beta$ in $\mathcal{X}'_n$, by considering two cases, recalling $f(\omega) =  w_1 \oplus w_2 \oplus w_n$:
\begin{itemize}
\item    If $x_{t} = f(\beta)$, then by the definition of $S$, neither $\beta$ nor $\hat \beta$ are in $\Samerep{n}$.
By the definition of $f$ and since $x_t \neq x_{t-1}$,  the first two bits of $\beta$ must differ from each other.  Thus, the run length of $\beta$ is one more than the run length of $\hat \beta$.
Thus the claim holds by Proposition~\ref{fact:tough}. 
\item If $x_{t} \neq f(\beta)$, then either $\beta$ or $\hat \beta$ are in $\Samerep{n}$.   Let $\beta = b_1b_2\cdots b_n$.
Then $\PRR(\beta) = b_2\cdots b_n s_{t}$  and $\PRR(\hat \beta) = b_2\cdots b_n x_{t}$.  
%From Lemma~\ref{lem:rle}, the strings $\beta$ and $b_2\cdots b_n s_{t}$ have the same run length and
%the strings $\hat \beta$ and $b_2\cdots b_n x_{t}$ have the same run length.  Since $b_2\cdots b_n x_{t}$  has run length one greater than that of $b_2\cdots b_n s_{t}$, it must be that $\hat \beta$ has run length one greater than that of $\beta$.  This means that $\hat \beta$ must begin with 10 or 01, and hence is not a same-rep, which can be inferred by definition.
Since $f(\beta) = b_n$, $b_1 = b_2$, which implies $\hat \beta$ is not in $\Samerep{n}$.
Thus $\beta$ is a same-rep and the claim thus holds by Observation~\ref{obs:tree} (item 1).
  \end{itemize}
Since $\hat \beta$ appears before $\beta$ in $\mathcal{X}'_n$ then $\hat \beta$ must be a substring of $alt(n{-}1) x_1\cdots x_{t-2}$.
Thus, either $x_{t-n+1}\cdots x_{t-1} x_{t}$ or $x_{t-n+1}\cdots x_{t-1}s_{t}$ must be in  $alt(n{-}1) x_1\cdots x_{t-1}$ which
contradicts the fact that both $\mathcal{X}_n$ and $\same{n}$ are de Bruijn sequences.  Thus, there is no $n < t \leq 2^n$ such that $s_t \neq x_t$
and hence  $\same{n} = \mathcal{X}_n$.
\end{proof}

%=========================================================================================
%=================================================================
%=================================================================
\section{A de Bruijn successor for $\opposite{n}$}  \label{sec:opposite}

To develop an efficient de Bruijn successor for $\opposite{n}$, we follow an approach similar to that for $\same{n}$, except this time we focus on the  lexicographically smallest RLEs and RL2-reps.  Again, we consider three different representatives for the cycles $\bR_1, \bR_2, \ldots , \bR_t$ of $\bB(n)$ induced by the $\PRR$. 

\begin{itemize}
\item \defo{RL2-rep}: The string with the lexicographically smallest RLE; if there are two such strings,  it is the one beginning with 0.
\item \defo{LC2-rep}:  
The strings $0^n$ and $1^n$ for the classes $\{0^n\}$ and $\{1^n\}$ respectively.  
%The RL-rep for cycles with run length 1 and $n$.
For all other classes, it is the string $\omega$ with RLE $r_1r_2\cdots r_{\ell}$ such that $r_1 = 1$ and 
$\PRR^{r_2}(\omega)$ is  the RL2-rep.
%$r_2$ applications of the PRR starting with $\omega$ yields the RL2-rep.
%The string $\omega$ such that $r$ applications of the PRR starting with $\omega$ yields the RL2-rep, where $r$ is last value in the RLE of the RL2-rep.
\item \defo{opp-rep}: 
 $\left\{ \begin{array}{ll}
        \mbox{RL2-rep} &\ \  \mbox{if the RL2-rep is \emph{opp-special}}\\
        \mbox{LC2-rep}\  &\ \  \mbox{otherwise.}\end{array} \right.$        
\end{itemize}      
We say an RL2-rep is \defo{opp-special} if it belongs to the set $\SpecialO(n)$ defined as follows:

\vspace{-0.05in}

\begin{quote}
 $\SpecialO(n)$ is the set of length $n$ binary strings that begin with 1 and have RLE of the form $1x^z y$ where $z$ is odd and $y>x$.
 \end{quote}%
\vspace{-0.05in}

\noindent
The RL2-reps have already been illustrated in Section~\ref{sec:feedback}.  There are relatively few strings in $\SpecialO(n)$ and they all have odd run length;  they belong to PCR-related cycles.  The need for identifying opp-special strings is revealed in the proof for the upcoming Proposition~\ref{fact:tough2}.

\begin{exam}   \small  \label{exam:opp-rep}  
The RLEs of the strings in $\SpecialO(n)$ for $n=10,11,12,13$:
\medskip

\noindent
$n=10$:  \ \ 111111112, 1111114, 11116, 118, 12223, 127, 136, 145  \\
$n=11$: \ \  111111113, 1111115, 11117, 119, 12224, 128, 137, 146  \\
$n=12$: \ \ 11111111112, 111111114, 1111116, 11118,11(10), 12225, 129, 138, 147, 156  \\
$n=13$: \ \ 11111111113, 111111115, 1111117, 11119,11(11), 12226, 12(10), 139, 148, 157  
\end{exam}

  Except for the cases $0^n$ and $1^n$, the LC-rep will begin with 10 and 01.  As an example, consider $\omega = 10000101001$
which has RLE $r_1r_2r_3r_4r_5r_6r_7 =  1411121$.  It is an LC-rep since %$r_2=4$ applications of the PRR to $\omega$ 
$\PRR^{4}(\omega)$ 
is the RL2-rep $01010010000$ with RLE $1111214$.  Note the
 last value of this RLE will correspond to $r_2$.

Let $\RLLrep{n}$,  $\LCCrep{n}$, and $\OPPrep{n}$  denote the set  of all  length $n$ RL2-reps, LC2-reps, and opp-reps, respectively,  {\bf not including} the representative $0^n$.   Consider the following feedback functions where $\omega = w_1w_2\cdots w_n$ and $f(\omega) =  w_1 \oplus w_2 \oplus w_n$:

\begin{center}
$RL2(\omega) = \left\{ \begin{array}{ll}
            \overline{f(\omega)} &\ \  \mbox{if $\omega$ or $\hat \omega$  is in $\RLLrep{n}$;}\\
        f(\omega) &\ \  \mbox{otherwise,}\end{array} \right.$

\medskip

$LC2(\omega) = \left\{ \begin{array}{ll}
            \overline{f(\omega)} &\ \  \mbox{if $\omega$ or $\hat \omega$  is in $\LCCrep{n}$;}\\
        f(\omega) &\ \  \mbox{otherwise,}\end{array} \right.$
        
        \medskip

$O(\omega) = \left\{ \begin{array}{ll}
            \overline{f(\omega)} &\ \  \mbox{if $\omega$ or $\hat \omega$  is in $\OPPrep{n}$;}\\
       	    f(\omega) &\ \  \mbox{otherwise.}\end{array} \right.$

\end{center}

\begin{theorem}  \label{thm:successorO}
The feedback functions $RL2(\omega)$, $LC2(\omega)$ and $O(\omega)$ are de Bruijn successors.
\end{theorem}
\begin{proof}
Let the partition $\bR_1, \bR_2, \ldots, \bR_t$ of $\bB(n)$ induced by the PRR be  listed such that $\bR_t = \{1^n\}$ and the remaining $t{-}1$ cycles are ordered in non-decreasing order with respect to the run length of each cycle.  This means that $\bR_1 = \{0^n\}$ and its representative, which must be $0^n$, is not in the sets $\RLLrep{n}$,  $\LCCrep{n}$, and $\OPPrep{n}$ by their definition.   Now consider $\bR_i$ for $1 < i < t$.
Clearly the RL2-rep for  $\bR_i$,  which is a string with the lexicographically smallest RLE, will begin with 01 or 10.    Similarly, the LC2-rep for $\bR_i$  must begin with 01 or 10 by its definition.      
Together these results imply that each opp-rep for $\bR_i$  will also begin with 01 or 10.
Thus, if follows directly from Theorem~\ref{thm:RL2} that $RL2(\omega)$, $LC2(\omega)$ and $O(\omega)$ are de Bruijn successors.
\end{proof}
%=====================

Recall from Proposition~\ref{fact:Opre} that the length $n$ suffix of $\opposite{n}$ is $10^{n-1}$.  Let $\mathcal{Y}_n = y_1y_2 \cdots y_{2^n}$  be the de Bruijn sequence returned by {\sc DB}($O, 10^{n-1})$; it will have suffix $10^{n-1}$.
 Let $\mathcal{Y}'_n$ denote the linearized de Bruijn sequence  $0^{n-1} \mathcal{Y}_n$.
  Our goal is to show that  $\mathcal{Y}_n = \opposite{n}$.   Our proof applies the following two propositions.

\begin{proposition} \label{fact:prefixO}
$\mathcal{Y}_n$ has length $n$ prefix $010101\cdots$.
\end{proposition}
\begin{proof}
The result follows from $n$ applications of the successor $O$ to the seed $10^{n-1}$.
\end{proof}

\begin{proposition} \label{fact:tough2}
%If  the run length of $\omega  = w_1w_2\cdots w_n$ is  one less than the run length of $\hat \omega$ and neither $\omega$ nor $\hat \omega$ are opp-reps,  then $\hat \omega$ comes before $\omega$ in $\mathcal{Y}'_n$.
%
If $\beta$ is a string in $\bB(n)$ such that the run length of $\beta$ is one less than the run length of $\hat \beta$ and neither $\beta$ nor $\hat \beta$ are opp-reps,  then $\hat \beta$ appears before $\beta$ in $\mathcal{Y}'_n$.
\end{proposition}

\noindent A proof of this proposition is given later in Section~\ref{sec:proof2}.

%------------------
\begin{result}
\vspace{-0.2in}
\begin{theorem} \label{thm:main2}
The de Bruijn sequences $\opposite{n}$ and $\mathcal{Y}_n$ are the same. 
\end{theorem}
\vspace{-0.2in}
\end{result}

\noindent
\begin{proof}
Let $\opposite{n}  = o_1o_2\cdots o_{2^n}$, let $\mathcal{Y}_n = y_1y_2\cdots y_{2^n}$.    
From Proposition~\ref{fact:Opre} and Proposition~\ref{fact:prefixO}, $y_1y_2\cdots y_n = o_1o_2\cdots o_n = 0101\cdots $ and  moreover $\opposite{n}$ and $\mathcal{Y}_n$ share the same length $n{-}1$ suffix $0^{n-1}$. 
Based on these prefix and suffix conditions and because both $\opposite{n}$ and $\mathcal{Y}_n$ are de Bruijn sequences, 
clearly the substring $01^{n-1}$ is followed by a 1 in both sequences. Suppose there exists some smallest $t$, where $n < t \leq 2^n$,  such that $o_t \neq y_t$.   Let $\beta = y_{t-n}\cdots y_{t-1}$ denote the length $n$ substring of $\mathcal{Y}_n$ ending at position $t{-}1$. 
Then $y_t = y_{t-1}$, because otherwise the RLE of $\mathcal{Y}_n$ is lexicographically smaller than that of $\opposite{n}$, contradicting Proposition~\ref{fact:rle2}.
We claim that $\hat \beta$ comes before $\beta$ in $\mathcal{Y}'_n$, by considering two cases, recalling $f(\omega) =  w_1 \oplus w_2 \oplus w_n$:
\begin{itemize}
\item    If $y_{t} = f(\beta)$, then by the definition of $O$, neither $\beta$ nor $\hat \beta$ are in $\OPPrep{n}$.
By the definition of $f$ and since $y_t = y_{t-1}$, the first two bits of $\beta$ are the same.  Thus, the run length of $\beta$ is one less than the run length of $\hat \beta$.
Thus the claim holds by Proposition~\ref{fact:tough2}. 
\item If $y_{t} \neq f(\beta)$, then either $\beta$ or $\hat \beta$ are in $\OPPrep{n}$.  
Let $\beta=b_1b_2\cdots b_n$.  
Then $\PRR(\beta) = b_2\cdots b_n o_{t}$ and  $\PRR(\hat \beta) = b_2\cdots b_n y_{t}$.  
%From Lemma~\ref{lem:rle}, this means $\beta$ and $b_2\cdots b_n o_{t}$ have the same run length and 
%$\hat \beta$ and $b_2\cdots b_n y_{t}$ have the same run length.
%Since $b_2\cdots b_n y_{t}$ has run length one less than that of $b_2\cdots b_n o_{t}$, it must be that $\hat \beta$ has run length one less than that of $\beta$.
 %This means $\hat \beta$ must begin with 00 or 11 and hence is not an opp-rep, which can be inferred by definition.  
 Since $f(\beta) \neq b_n$, $b_1 \neq b_2$, which implies $\hat \beta$ is not in $\OPPrep{n}$ since the case when $\beta \neq 01^{n-1}$ was already handled.
 Thus $\beta$ is an opp-rep and the claim holds by Observation~\ref{obs:tree} (item 1).
  \end{itemize}
Since $\hat \beta$ appears before $\beta$ in $\mathcal{Y}'_n$ then $\hat \beta$ must be a substring of $0^{n-1} y_1\cdots y_{t-2}$.
Thus, either $y_{t-n+1}\cdots y_{t-1} y_{t}$ or $y_{t-n+1}\cdots y_{t-1}o_{t}$ must be in  $0^{n-1} y_1\cdots y_{t-1}$ which
contradicts the fact that both $\mathcal{Y}_n$ and $\opposite{n}$ are de Bruijn sequences.  Thus, there is no $n < t \leq 2^n$ such that $o_t \neq y_t$
and hence  $\opposite{n} = \mathcal{Y}_n$.
\end{proof}

%=========================================================================================
%=================================================================
%=================================================================
\section{Lexicographic compositions}  \label{sec:LC}

As mentioned earlier, Fredricksen and Kessler  devised a construction based on lexicographic compositions~\cite{lexcomp}.    Let $\LC{n}$ denote  the de Bruijn sequence of order $n$ that results from this construction.
The sequences $\same{n}$ and $\LC{n}$ first differ at $n=7$ (as noted below), and for $n\geq 7$ they were conjectured to match for a significant prefix~\cite{fred-nfsr,lexcomp}:

\begin{center} \small
$\same{7} = \begin{array}{l}
        1111111000000011111011110011110100000100001100001011100011100100\\
        0110111011000100111010110011001011011010\underline{\red{011010100010100100101010}},
        \end{array}  $
        
        \medskip
        
$\LC{7} = \begin{array}{l}
        1111111000000011111011110011110100000100001100001011100011100100\\
        0110111011000100111010110011001011011010\underline{\blue{100010100110100100101010}}.
        \end{array}  $      
\end{center}

After discovering the de Bruijn successor for $\same{n}$, we observed that the de Bruijn sequence resulting from the de Bruijn successor $LC(\omega)$ corresponded to $\LC{n}$ for small values of $n$.   Recall that $alt(n)$ denotes the alternating sequence of 0s and 1s of length $n$ that ends with 0.   Let $\mathcal{LC}_n$  be the de Bruijn sequence returned by {\sc DB}($LC, 0 alt(n{-}1)$). 
  
\begin{conjecture} \label{conj:lc}
The de Bruijn sequences $\mathcal{LC}_n$ and $\LC{n}$ are the same. 
\end{conjecture}

We verified that $\mathcal{LC}_n$ is the same as $\LC{n}$ for all $n<30$.  However, as the description of the algorithm to construct $\LC{n}$ is rather detailed~\cite{lexcomp}, we did not attempt to prove this conjecture.

%=========================================================================================
%=================================================================
%=================================================================
\section{Efficient implementation} \label{sec:implement}

Given a membership tester for $\RLrep{n}$, testing whether or not a string is an LC-rep or a same-rep can easily be done in $O(n)$ time and $O(n)$ space.
Similarly, given the membership tester for $\RLLrep{n}$, testing whether or not a string is an LC2-rep or a opp-rep can easily be done in $O(n)$ time and $O(n)$ space.
Thus, by applying Proposition~\ref{fact:RL-test} and Proposition~\ref{fact:RL2-test}, we can implement each of our six de Bruijn successors in $O(n)$ time using $O(n)$ space.

\begin{theorem}
The six de Bruijn successors 
 $RL(\omega)$, $LC(\omega)$, $S(\omega)$ , 
$RL2(\omega)$, $LC2(\omega)$ and $O(\omega)$
can be implemented in $O(n)$ time using $O(n)$ space.
\end{theorem}

%============================================
%============================================
%============================================
%============================================
%============================================
%============================================

\section{Proof of Proposition~\ref{fact:tough}} \label{sec:proof1}

Recall that  $\mathcal{X}_n$ = {\sc DB}($S, 0 alt(n{-}1$)) and $\mathcal{X}'_n = alt(n{-}1)\mathcal{X}_n$.
We begin by restating Proposition~\ref{fact:tough} by reversing the roles of $\beta$ and $\hat \beta$ in the original statement for convenience:
\begin{quote}
If $\beta$ is a string in $\bB(n)$ such that the run length of $\beta$ is one less than the run length of $\hat \beta$ and neither $\beta$ nor $\hat \beta$ are same-reps,  then $\beta$ appears before $\hat \beta$ in $\mathcal{X}'_n$.
\end{quote}
The first step is to further refine the ordering of the cycles $\bR_1, \bR_2, \ldots , \bR_t$ used in the proof of Theorem~\ref{thm:successorS} to prove that $S(\omega)$ was a de Bruijn successor.   In particular, let $\bR_1, \bR_2, \ldots , \bR_t$ be the cycles of $\bB(n)$ induced by the $\PRR$ ordered in non-increasing order with respect to the run lengths of each cycle,  \emph{\blue{additionally refined so the cycles with the same run lengths are ordered in decreasing order with respect to the RLE of the RL-rep}}.  If two RL-reps have the same RLE, then the cycle with RL-rep starting with 1 comes first. 
Let $\sigma_i, \gamma_i, \alpha_i$ denote the RL-rep, LC-rep, and same-rep, respectively,  for $\bR_i$, where $1 \leq i \leq t$;  let $R_i$ denote the RLE of $\sigma_i$.

Assume the run length of $\beta$ is one less than the run length of $\hat \beta$ (the RLE of $\beta$ must begin with a value greater than~1), and neither $\beta$ nor $\hat \beta$ are same-reps. 
Since each string in $\bR_1$ has maximal run length $n$, $\beta \in \bR_i$ for some $1 < i \leq t$ and thus $R_i$ is of the form $r_1r_2\cdots r_m1^v$ where $r_m > 1$.
Let $\bR_j$ contain $\hat \alpha_i$ which means $\bR_j$ is the parent of $\bR_i$.    
Let $\bR_k$ contain $\hat \beta$.  In general, we will show that either $j<k$ or $j=k$; 
see Figure~\ref{fig:beta}.  
The cases for when $j<k$ are handled in Section~\ref{sec:jk}.   In the next steps, we will focus on the situations when $j=k$.  
 Through computer experimentation for $n \leq 25$, we verified that $j=k$ only for specific instances of  $\beta$ equal to $\gamma_i$, $\overline{\gamma_i}$, $\sigma_i$, or $\overline{\sigma_i}$.
In our formal proof, we find that $\bR_j$ is aperiodic.  Thus, by Observation \ref{obs:tree} (item 2), we determine
the smallest positive integers $a$ and $b$ such that $\PRR^a(\hat \alpha_i) = \hat \beta$ and 
$\PRR^b(\hat \alpha_i) = \alpha_j$ and demonstrate that $a<b$.  

\begin{result}
\noindent {\bf Outline of next steps:}

\begin{enumerate}
\item $\sigma_i \in \Special(n)$ 

\item $\sigma_i \notin \Special(n)$ 
	\begin{itemize}
		\item Consider $m=1$ 
		\item Consider $m>1$ 
		\begin{itemize}
			\item Handle the case when $\beta = \overline{\gamma_i}$
			\item Consider one RLE possibility for $\beta$ which includes an instance when $\beta = \overline{\gamma_i}$
			\item Consider a second RLE possibility for $\beta$ which includes instances when $\beta = \sigma_i$, $\beta = \overline{\sigma_i}$, and $\beta = \overline{\gamma_i}$
		        \item Handle the instances when $\beta = \sigma_i$ or $\beta = \overline{\sigma_i}$ 
		\end{itemize}
	\end{itemize}
\end{enumerate}			

\end{result}

%---------------------
\begin{figure}[h]
\begin{center}
\resizebox{4.1in}{!}{\includegraphics{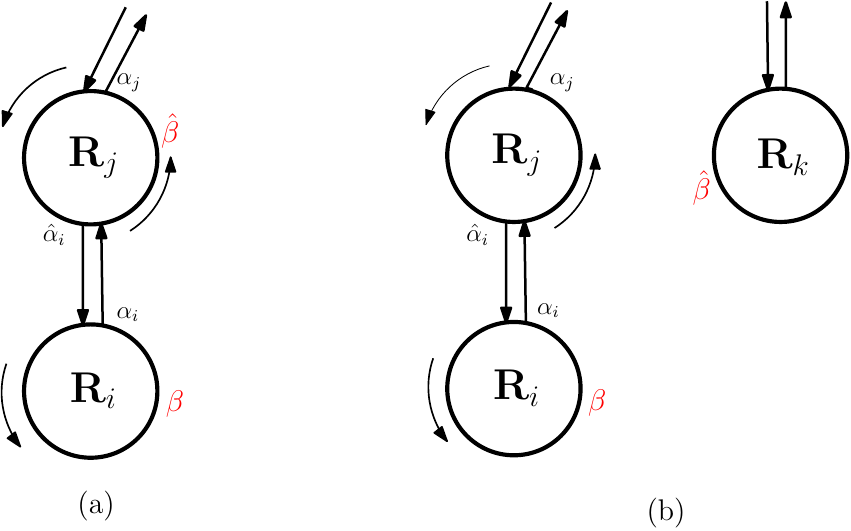}}
\end{center}

\vspace{-0.15in}
%\captionsetup{format=hang}
\caption{  \small
Illustrating the possible relationships between $\beta$ and $\hat \beta$: 
(a) $j=k$,    (b) $j<k$, noting the run length of $\bR_j$ and $\bR_k$ are the same. }
\label{fig:beta}
\end{figure}
%---------------------

%==============================
% SPECIAL
%==============================
\smallskip

\noindent
{\bf CASE 1: $\sigma_i \in \Special$.} ~
In this case, $\alpha_i = \sigma_i$ has RLE of the form $(21^{2x})^y1^z$ and begins with 0, where $x \geq 0$ and $y,z \geq 2$. 
Thus $\hat \alpha_i \in \bR_j$ has RLE $1^{2x+2}(21^{2x})^{y-1}1^z$.
Considering the RLE possibilities of the other strings in $\bR_j$, as outlined in Lemma~\ref{lem:RLEs},  we deduce
\[  \sigma_j = \PRR^{2x+2}(\hat \alpha_i)  \mbox{ begins with 1 and has RLE } (21^{2x})^{y-1}1^{z+2x+2}. \]
Clearly $\bR_j$ is aperiodic.
Suppose $\beta = \gamma_i$; it will have RLE  $21^{z+2x-1}(21^{2x})^{y-1}1$.
Observe that $\PRR^{z+2x+1} (\hat \beta)$ has the same RLE as $\sigma_j$, but begins with 0. 
Thus, since $\bR_j$ is CCR-related and applying Observation~\ref{obs:PRR}, $\PRR^{(z+2x+1) + (n-1)}(\hat \beta) = \sigma_j$, and thus $\PRR^{(n-1) - z + 1}(\hat \alpha_i)  = \hat \beta$.  
By definition, $\PRR^{z+2x+1+2x+2}(\gamma_j) = \sigma_j$, which means that $\PRR^{z+2x+1}(\gamma_j)  = \hat \alpha_i$.   
Since $\bR_j$ is CCR-related, $\alpha_j = \gamma_j$.  Thus $a = n-z$, $b = (2n-2)-(z+2x+1)$, and clearly $a<b$.
For all other cases such that $\beta \neq \sigma_i$ (the same-rep), it is a simple exercise to see that $R_j > R_k$, and hence $j<k$. 

\begin{exam} \small
Consider $\bR_i$ where  $\alpha_i = \sigma_i = 00101101010 \in \Special(n)$ and has RLE 211211111. The corresponding LC-rep  $\beta = \gamma_i = 11010100101$ has RLE 211112111.  Below are the strings from $\bR_j$
including $\hat \sigma_i$ and $\hat \gamma_i$ in the order that they 
 appear in  $\mathcal{X}'_{11}$.   Note that $\hat \beta$ appears after $\hat \alpha_i$ ($a=8, b=14$).

\medskip

\begin{tabular} {l}
01010101011 \\
10101010110 \\
01010101101 \\
10101011010 \\
01010110101 \\
10101101010 \ \ $\leftarrow$  \ \   \red{$\hat \sigma_i = \hat \alpha_i$} \\
01011010101 \\
10110101010 \\
01101010101 \\
11010101010  \ \ $\leftarrow$   \ \    $\sigma_j$, the RL-rep, with RLE 2111111111 \\
10101010100 \\
01010101001 \\
10101010010 \\
01010100101 \ \ $\leftarrow$  \ \   \red{$\hat \gamma_i = \hat \beta$} \\
10101001010 \\
01010010101\\
10100101010 \\
01001010101 \\
10010101010 \\
00101010101  \ \ $\leftarrow   \ \ \alpha_j = \gamma_j (= \overline{\sigma_j})$, the same-rep and LC-rep for this cycle 
\end{tabular}

\end{exam}
%+====================

%==============================
% NOT SPECIAL
%==============================
\smallskip

\noindent
{\bf CASE 2: $\sigma_i \notin \Special(n)$}. ~ By definition $\alpha_i = \gamma_i$.  This case involves some rather technical analysis of the RLE for various strings. Assume $R_i = r_1r_2\cdots r_m1^v$, where $m \geq 1$ and $r_1,r_m \geq 2$.  Then,
\begin{itemize}
\item $\alpha_i = \gamma_i$ has RLE  $21^{v-1}r_1r_2\cdots r_{m-1}(r_m{-}1)$ where  $\PRR^{v+1}(\alpha_i) = \sigma_i$, and
\item $\hat \alpha_i$ has RLE $1^{v+1}r_1r_2\cdots r_{m-1}(r_m{-}1)$ and is in $\bR_j$.
\end{itemize}
Consider the RLE possibilities of the other strings in $\bR_j$ as outlined in Lemma~\ref{lem:RLEs}.  Given that $\sigma_i$ is an RL-rep, we deduce 
\begin{center}
$\sigma_j = \left\{ \begin{array}{ll}
        \PRR^{v+1}(\hat \alpha_i) &\ \  \mbox{if $\sigma_i$ begins with 1;}\\
        \PRR^{(n-1) + v+1}(\hat \alpha_i)  &\ \  \mbox{if $\sigma_i$ begins with 0 ($\bR_i$ is PCR-related).}\end{array} \right.$
\end{center}
In both cases $\sigma_j$ begins with 1 (implying $\alpha_j = \gamma_j$) and  $R_j =  r_1r_2\cdots r_{m-1}(r_m{-}1)1^{v+1}$.
\begin{claim} \label{claim:parent}
If $\bR_j$ is the parent of $\bR_i$ then $\sigma_j$ begins with 1 and $\bR_j$ is aperiodic.  
\end{claim} 
Note this claim also held for the case when $\sigma_i \in \Special(n)$.  Observe that $\bR_j$ is indeed aperiodic, since if we assume otherwise, it implies that $\sigma_i$ is not 
an RL-rep.

\smallskip
%===============
% m = 1
%===============
{\bf  Suppose $m=1$.}  Then the RLE of $\beta$ is $r_1 1^v$, the RLE of $\hat \beta$ is $1(r_1{-}1)1^v$,  the RLE of $\alpha_i = \gamma_i$ is  $21^v(r_1{-}1)$ and 
the RLE of $\hat \alpha_i$ is  $1^{v+1}(r_1{-}1)$.  Thus, $R_j = (r_1{-}1)1^{v+1}$ and the RLE  of $\gamma_j$ is $21^{v}(r_1{-}2)$.
If $\bR_i$ is CCR-related and $\beta = \overline{\sigma_i}$, which begins with 0, then $R_j = R_k$ where $\sigma_j$ begins with 1 and $\sigma_k$ begins with 0.  Thus $j<k$.  
Otherwise,  $\beta = \sigma_i$.  
By its RLE, $\sigma_j \notin \Special(n)$, so $\alpha_j = \gamma_j$.
From the above RLEs,  $\PRR^{v+1}(\hat \alpha_i) = \hat \beta$ and thus $a=v+1$.   
Applying Observation~\ref{obs:PRR}, if $\bR_j$ is a CCR-related cycle and $\beta$ begins with 1, then it is easily verified that $b= (2n-2)-1$ is the smallest value such that $\PRR^{b}(\hat \alpha_i) = \alpha_j$; otherwise $b=(n-1)-1$ is the smallest value such that $\PRR^{b}(\hat \alpha_i) = \alpha_j$.  In both cases $a<b$.

\smallskip
%===============
% m > 1
%===============
{\bf Suppose $m>1$. }
Let $d = m$, unless $r_m = 2$, in which case let $d$ be the largest index less than $m$ such that $r_d > 1$.   Then given $R_j$, $\sigma_j = \PRR^{m-d+1+ (v+1)}(\gamma_j).$
Thus: 
\begin{center}
\begin{equation} \label{eq:sameb}
\alpha_j = \gamma_j = \left\{ \begin{array}{ll}
        \PRR^{(n-1) - (m-d+1)}(\hat \alpha_i) &\ \  \mbox{if $\bR_i$ is CCR-related;}\\
         \PRR^{(2n-2) - (m-d+1)}(\hat \alpha_i) &\ \  \mbox{if $\bR_i$ is PCR-related and $\sigma_i$ begins with 1;} \\
        \PRR^{(n-1) - (m-d+1)}(\hat \alpha_i) &\ \  \mbox{if $\bR_i$ is PCR-related and $\sigma_i$ begins with 0.}
        \end{array} \right.
\end{equation}
\end{center}
Consider $\beta = \overline{\gamma_i}$.  
 If $\bR_i$ is CCR-related then $\sigma_j$ begins with 1, but $\sigma_k$ begins with 0, and hence $j<k$.
Otherwise, $\bR_i$ is PCR-related  and $\bR_j$ is CCR-related and hence $\alpha_j = \gamma_j$.   Since both $\gamma_i$ and $\overline{\gamma_i}$ belong to $\bR_i$,  both $\sigma_i$ and $\overline{\sigma_i}$  belong to $\bR_i$.  Thus $\sigma_i$ begins with 1 and $b = (2n{-}2) - (m{-}d{+}1) $.
 Since $\alpha_i = \gamma_i$,  we have
$\PRR^{n-1}(\hat \alpha_i) =  \overline{\hat \alpha_i} = \hat \beta$.   
%As noted earlier, $\sigma_j = \PRR^{v+1}(\hat \alpha_i)$  and $\sigma_j = \PRR^{m-d+1+ (v+1)}(\gamma_j)$ 
%which implies $\hat \alpha_i = \PRR^{m-d+1}(\gamma_j)$.  
Thus, $a=n{-}1$ and clearly $a<b$.

Based on the possibilities given in Lemma~\ref{lem:RLEs} using $\omega = \sigma_i$, the RLE for other possible $\beta$ is of the form 
\begin{equation*} 
(r_s{-}j)r_{s+1}\cdots r_m 1^{v-1} r_1\cdots r_{s-1} (j{+}1),$$ 
\end{equation*}
for some $1\leq s \leq m$ where $r_s \geq 2$ and $0 \leq j \leq r_s-2$.  Thus, $\hat \beta$ has RLE of the form
\begin{equation*} \label{eq:hat}
1(r_s{-}j{-}1)r_{s+1}\cdots r_m 1^{v-1} r_1\cdots r_{s-1} (j{+}1).$$ 
\end{equation*}
Note that $r_1 \geq r_q$ for all $1 < q \leq m$.  If $R_k$ begins with a value less than $r_1$, then clearly $R_j > R_k$ and $j<k$. 
Otherwise, based on the possible RLEs for $\hat \beta$  and applying Lemma~\ref{lem:RLEs} (using $\omega = \hat \beta$), 
$R_k$ must begin with some $r_{s'} = r_1$ and have the form
\small
\begin{eqnarray}  
   &&  r_{s'} \cdots r_{s-1} \blue{(j{+}1)  (r_s{-}j{-}1) }  r_{s+1} \cdots r_m1^{v-1}  r_1\cdots r_{s'-1}1,   \mbox{\ \ such that  $0 < s' < s$,  or }  \label{eq:rk1} \\   
    && r_{s'} \cdots r_m 1^{v-1} r_1\cdots r_{s-1}  \blue{(j{+}1)  (r_s{-}j{-}1) }    r_{s+1}\cdots r_{s'-1}1,   \mbox{\ \ such that  $s < s' \leq m$,  }   \label{eq:rk2} 
\end{eqnarray}
where  $0\leq j \leq r_s-2$.
\normalsize 

%-------------------------------------------
Suppose $R_k$ has the form in (\ref{eq:rk1}).   Since $\sigma_i$ is an RL-rep, 
\[ r_1\cdots r_m1^{v}   \ \  \geq  \ \    r_{s'} \cdots r_{s-1} \blue{r_s}  r_{s+1} \cdots r_m1^{v-1}  r_1\cdots r_{s'-1}1.  \] 
We want to compare
\begin{eqnarray*} 
R_j  & =  &  r_1r_2\cdots r_{m-1}(r_m{-}1)1^{v+1} \  \ \ \mbox{ with \ \ \ }   \\
R_k & =   &  r_{s'} \cdots r_{s-1} \blue{(j{+}1) (r_s{-}j{-}1) }  r_{s+1} \cdots r_m1^{v-1}  r_1\cdots r_{s'-1}1.
\end{eqnarray*}
Clearly $R_j > R_k$ unless $r_{s'}\cdots r_{s-1}(j{+}1)= r_1\cdots r_{m-1}(r_m{-}1)$ and $(r_s{-}j{-}1)  r_{s+1} \cdots r_m1^{v-1}  r_1\cdots r_{s'-1}1 = 1^{v+1}$.
Thus, $s=m$, $s' = 1$, and $j = r_m-2$.   This implies $\beta = \gamma_i$ (same-rep) or $\beta = \overline{\gamma_i}$.

%-------------------------------------------
Suppose $R_k$ has the form in (\ref{eq:rk2}).  Since $\sigma_i$ is an RL-rep, 
\[ r_1\cdots r_m1^{v} \ \  \geq \ \     r_{s'} \cdots r_m 1^{v-1} r_1\cdots r_{s-1}  \blue{r_s}    r_{s+1}\cdots r_{s'-1}1 .  \] 
We want to compare
\begin{eqnarray*} 
R_j  & =  &  r_1r_2\cdots r_{m-1}(r_m{-}1)1^{v+1} \  \ \ \mbox{ with \ \ \ }   \\
R_k & =   &   r_{s'} \cdots r_m 1^{v-1} r_1\cdots r_{s-1}  \blue{(j{+}1) (r_s{-}j{-}1) }    r_{s+1}\cdots r_{s'-1}1.
\end{eqnarray*}
 Clearly $R_j \geq R_k$ based on the RLEs described before this claim.
 Suppose $R_j =  R_k$.  Then $s'-s \leq v$, and the suffix  $\blue{(r_s{-}j{-}1) } r_{s+1}\cdots r_{s'-1}1$ must be all 1s with
$(j+1) = r_{m}-1$.  Thus $j=r_s-2$ and $r_s=r_m$. 
 If $s'-s = v$, then the RLE for $\beta$ is the same as $R_i$.  Since $\sigma_i$ is an RL-rep, this implies that $r_1r_2\cdots r_m$ is periodic.  This is easily deduced 
 since a proper suffix of $r_1r_2\cdots r_m$  is equal to a prefix.
If $r_1r_2\cdots r_m$ is not of the form $(21^p)^q$ for $p\geq 0$ and $q>1$, then $R_k < R_j$, contradiction.  This means $\beta = \gamma_i$ (same-rep), or $\beta = \overline{\gamma_i}$  (already handled).
 If $s'-s < v$ then clearly $r_s = r_m =  2$, since both $\blue{(j{+}1)}$ and $\blue{(r_s{-}j{-}1)}$ must be 1.   
 Suppose $s \neq 1$.  Note that $r_{m-s+2} \cdots r_{m-1}(r_m-1)1^{v+1} = r_1\cdots r_{s-1}  \blue{11}    r_{s+1}\cdots r_{s'-1}1$.  
 However since $s'-s < v$, this implies that $r_1\cdots r_{s-1} < r_{m-s+2} \cdots r_m$ which contradicts the fact 
 that $\sigma_i$ is the RL-rep (applying Lemma~\ref{lem:RLEs}).  
 Thus $s=1$ and $r_1 = 2$.   Again comparing $R_j$ and $R_k$:
 \begin{eqnarray*} 
 R_j  & =  &   (r_1 \cdots r_{s'-1}) \ \blue{(r_{s'} \cdots r_{2s'-1})} \ \cdots \ (r_{m-s'} \cdots r_{m-1}) 1^{v+2} \  \ \ \mbox{ and \ \ \ }   \\
 R_k & =   &   \blue{(r_{s'} \cdots r_{2s'-1})} \   \cdots \ (r_{m-s'} \cdots r_{m-1}) r_m 1^{v-1} 1^{s'+1}.
 \end{eqnarray*}
 Since $r_1 \cdots r_{s'-1} = 21^{s'-2}$ and $r_m=2$, $\beta$ has RLE of the form $(21^{p})^q1^u$ where: 
 (i) $p \geq 0$ since $p=s'-2$;
 (ii)  $q>1$ since $m>1$;  and 
 (iii) $u  > 1$ since $u=(v+2) - (2 + s'-2) = v-s'+2$ and $s'-1 < v$.
%
%===============
% RLE 211121112  1111
%===============
Thus:  
\begin{itemize}
\item $\alpha_i = \gamma_i$ has RLE $21^{u+p-1} (21^p)^{q-1}1$,
\item $\hat \alpha_i$ has RLE $1^{u+p+1} (21^p)^{q-1}1$,
\item $\PRR^{u-1}(\hat \alpha_i)$ is $\hat \beta$ or $\overline{\hat \beta}$ ,  and
\item $d=m-(p+1)$  and $m-d+1 = p+2$.
\end{itemize}

%PCR-related
Suppose $\bR_i$ is PCR-related.  Then $\bR_j$ is CCR-related cycle and thus $\alpha_j = \gamma_j$.
If $\beta$ begins with 1, then $\beta = \sigma_i$ and $a = u-1$, and $b=(2n-2) - (p+2)$.  Clearly $a < b$.
If $\beta$ begins with 0 and $p$ is odd, then $\beta = \sigma_i$, $a=u-1$, and $b = (n-1) - (p+2)$ from (\ref{eq:sameb}).  Clearly $a<b$.   
If $\beta$ begins with 0 and $p$ is even, then if $u > 1$, $\sigma_i \in \Special(n)$, contradiction.

%CCR-related
Suppose $\bR_i$ is CCR-related.   Then $b = (n-1) - (p+2)$  from (\ref{eq:sameb}).
Suppose $\beta$ begins with 1;  $\beta = \sigma_i$.  If $p$ is even, then $R_j = R_k$, but $\sigma_j$ begins with 1 and $\sigma_k$ begins with 0.  Thus $j<k$.  If $p$ is odd, then $\sigma_j = \PRR^{p+2}(\hat \beta)$ and since $v=u+p$, we have $a=u+p+1 - (p+2) = u-1$ and thus $a<b$.  
Suppose $\beta$ begins with 0; $\beta = \overline{\sigma_i}$.  
If $p$ is odd, then $R_j = R_k$, but $\sigma_j$ begins with 1 and $\sigma_k$ begins with 0, and hence $j<k$.
If $p$ is even, then $\sigma_j = \PRR^{p+2}(\hat \beta)$ and again $a = u+p+1 - (p+2) = u-1$ and $a<b$.

%It remains to prove the cases for when $j<k$.

%===========================
%===========================
\subsection{ $j<k$ } \label{sec:jk}

The proof for the case when $j<k$ applies the following two claims.

%========================
\begin{claim} \label{claim:jk}
If $\bR_j$ and $\bR_k$ have the same run length where $j<k$ such that $\sigma_j$ and $\sigma_k$  both begin with~1, then every string from $\bR_j$ appears in $\mathcal{X}'_n$ before any string from $\bR_k$.
\end{claim}
%========================
\begin{proof}
The proof is  by induction on the levels of the related tree of cycles rooted by $\bR_1$.  The base case trivially holds for cycles with run length $n$ since there is only one such cycle $\bR_1$.  Assume that the result holds for all cycles at levels with run length greater than $\ell < n$.
Consider two cycles $\bR_j$ and $\bR_k$ with run length $\ell$ such that $\sigma_j$ and $\sigma_k$  both begin with~1; neither $\sigma_j$ nor $\sigma_k$ are same-special and since $j < k$,   $R_j > R_k$.
 Let $\bR_x$ and $\bR_y$ denote the parents of $\bR_j$ and $\bR_k$, respectively.  By Claim~\ref{claim:parent}, both $\sigma_x$ and $\sigma_y$ begin with 1.
 Given $R_j > R_k$,  our earlier analysis (just before Claim~\ref{claim:parent}) implies that the RLE of $\sigma_x$ is greater than the RLE of $\sigma_y$.
Thus, by the ordering of the cycles, $x < y$. 
By induction, every string from $\bR_x$ appears before  every string from $\bR_y$ in $\mathcal{X}'_n$, and hence by Observation~\ref{obs:tree} (item 4), we have our result.  
\end{proof}

%========================
\begin{claim} \label{claim:same}
Let $\bR_k$ and $\bR_{k'}$ be cycles with $k'< k$ such that $\sigma_k$ and $\sigma_{k'}$ have the same RLE  $r_1r_2\cdots r_m1^v$ where $r_m>1$ and $v \geq 0$.
Then every string from $\bR_{k'}$ appears in $\mathcal{X}'_n$ before any string from $\bR_k$.
\end{claim}
%========================
%
\begin{proof}
By the ordering of the cycles, $\sigma_{k'}$ begins with 1 and $\sigma_k$ begins with 0; they belong to PCR-related cycles.
Note that $\sigma_k = \overline{\sigma_{k'}}$ and similarly $\gamma_k = \overline{\gamma_{k'}}$.  Thus, $\hat \sigma_k = \overline{\hat \sigma_{k'}}$   and $\hat \gamma_k = \overline{\hat \gamma_{k'}}$ and each pair, respectively, will belong to the same CCR-related cycle.  
If $\sigma_k \in \Special(n)$, we previously observed that $\hat \gamma_k$ and $\hat \sigma_k$ belong to the same cycle, and thus $\bR_k$ and $\bR_{k'}$ have the same parent.
If $\sigma_k \notin \Special(n)$, then $\bR_k$ and $\bR_{k'}$ also have the same parent containing both $\hat \gamma_k$ and $\hat \gamma_{k'}$.
Let $\bR_{\ell}$ be the shared parent of $\bR_k$ and $\bR_{k'}$.    
Since $\sigma_{k'}$ begins with 1,  $\alpha_{k'} = \gamma_{k'}$.   
If $m=1$, 
then we already saw that 
$\alpha_{\ell} =  \PRR^{(2n-3)}(\hat \alpha_{k'})$ and 
$\alpha_{\ell} =  \PRR^{(n-2)}(\hat \alpha_{k})$.
If $m>1$, we observed that  
$\alpha_{\ell} =  \PRR^{(2n-2) - (m-d+1)}(\hat \alpha_{k'})$  from (\ref{eq:sameb}), 
recalling $d=m$ unless $r_m = 2$, in which case let $d$ is the largest index less than $m$ such that $r_d > 1$.
If $\sigma_k \in \Special(n)$ then $\alpha_k = \sigma_k$ and from earlier analysis $\alpha_{\ell} =  \PRR^{(2n-2) - (v+1)}(\hat \alpha_k)$ noting $(m-d+1) \leq v$ in this case; otherwise,  $\alpha_k = \gamma_k$, 
 and $\alpha_{\ell} =  \PRR^{(n-1) - (m-d+1)}(\hat \alpha_k)$  from (\ref{eq:sameb}).  
In all cases, applying Observation~\ref{obs:tree},  $\hat \alpha_{k'}$ appears before $\hat \alpha_k$ in $\mathcal{X}'_n$,
and every string in $\bR_{k'}$ appears in $\mathcal{X}'_n$ before any string in $\bR_k$.
\end{proof}
%=======
  
Recall that $\sigma_j$ begins with 1 from Claim~\ref{claim:parent}.  Thus, if $\sigma_k$ begins with 1, then Claim~\ref{claim:jk} implies that all strings from $\bR_j$ appear  in $\mathcal{X}'_n$ before all strings from $\bR_k$.  
Otherwise, if $\sigma_k$ begins with 0, then it must correspond to a PCR-related cycle.  
Consider $\bR_{k'}$ containing RL-rep $\overline{\sigma_k}$ which begins with 1; it has the same RLE as $\sigma_k$. 
By Claim~\ref{claim:same}, all strings from $\bR_{k'}$ appear in $\mathcal{X}'_n$  before all strings from $\bR_{k}$.  If $j=k'$, we are done; 
otherwise $j<k'< k$ and  Claim~\ref{claim:jk} implies that all strings from $\bR_j$ appear in $\mathcal{X}'_n$ before all strings from $\bR_{k'}$.
Finally, by applying Observation~\ref{obs:tree} (item 4), all strings from $\bR_i$ including $\beta$ will appear in $\mathcal{X}'_n$ before all strings from $\bR_k$ including $\hat \beta$.

%============================================
%============================================
%============================================
\section{Proof of Proposition~\ref{fact:tough2}} \label{sec:proof2}

The proof of this proposition follows  similar steps as the proof for Proposition~\ref{fact:tough}; however, the RLE analysis is less complex.
Recall that  $\mathcal{Y}_n$ = {\sc DB}($O, 10^{n-1}$) and $\mathcal{Y}'_n = 0^{n-1}\mathcal{Y}_n$.
We restate Proposition~\ref{fact:tough2}, reversing the roles of $\beta$ and $\hat \beta$ from its original statement for convenience:

\begin{quote}
If $\beta$ is a string in $\bB(n)$ such that the run-length of $\beta$ is one more than the run-length of $\hat \beta$ and neither $\beta$ nor $\hat \beta$ are opp-reps,  then $\beta$ appears before $\hat \beta$ in $\mathcal{Y}'_n$.
\end{quote}
The first step is to further refine the ordering of the cycles $\bR_1, \bR_2, \ldots , \bR_{t}$ used in the proof of Theorem~\ref{thm:successorO} to prove that $O(\omega)$ was a de Bruijn successor.   In particular, let $\bR_1, \bR_2, \ldots , \bR_{t-1}$ be the cycles of $\bB(n)$ induced by the $\PRR$, not including $\bR_t = \{1^n\}$,  ordered in non-decreasing order with respect to the run lengths of each cycle. 
\emph{\blue{This ordering is additionally refined so the cycles with the same run lengths are ordered in increasing order with respect to the RLE of the RL2-rep}}.    If two RL2-reps have the same RLE, then the cycle with RL2-rep starting with 0 comes first. 
Let $\sigma_i, \gamma_i, \alpha_i$ denote the RL2-rep, LC2-rep, and opp-rep, respectively,  for $\bR_i$, where $1 \leq i \leq t$;  let $R_i$ denote the RLE of $\sigma_i$.
Assume the run length of $\beta$ is one more than the run length of $\hat \beta$, and neither $\beta$ nor $\hat \beta$ are opp-reps. 
This run-length constraint implies that the RLE of $\beta$ must begin with 1.
Since each string in $\bR_1$ and $\bR_t$ has run length $1$, $\beta \in \bR_i$ for some $1 < i < t$.
Let $\bR_j$ contain $\hat \alpha_i$ which means $\bR_j$ is the parent of $\bR_i$.    
Let $\bR_k$ contain $\hat \beta$.  Like the proof in the previous section, we show that either $j<k$ or $j=k$; 
see Figure~\ref{fig:beta}.
The cases for when $j<k$ are handled in Section~\ref{sec:jk2}.   
As we analyze the cases when $j=k$, we find that $\bR_j$ is aperiodic.  Thus, by Observation \ref{obs:tree} (item 2), we determine
the smallest positive integers $a$ and $b$ such that $\PRR^a(\hat \alpha_i) = \hat \beta$ and 
$\PRR^b(\hat \alpha_i) = \alpha_j$ and demonstrate that $a<b$.  
% The remainder of the proof follow the same steps as the previous section.

%
%==============================
% SPECIAL
%==============================
\smallskip

\noindent
{\bf CASE 1: $\sigma_i \in \SpecialO(n)$}. ~
In this case $\alpha_i = \sigma_i$ begins with 1 and $R_i = 1x^zy$ where $z$ is odd and $y>x$. 
Thus $\hat \alpha_i \in \bR_j$ begins with 0 and has RLE $(x{+}1)x^{z-1}y$.
Considering the RLE possibilities of the other strings in $\bR_j$, as outlined in Lemma~\ref{lem:RLEs},  clearly
%$\sigma_j$ begins with 0, since $\bR_j$  is CCR-related, and $R_j = 1x^{z-1}(y+x)$. Clearly $\bR_j$ is aperiodic.
%
\[  \sigma_j = \PRR^{x}(\hat \alpha_i)  \mbox{ begins with 0 and has RLE } 1x^{z-1}(y{+}x), \]
and $\bR_j$ is aperiodic.
Suppose $\beta = \gamma_i$; it will have RLE  $1yx^z$ and begin with 0.
Observe that $\PRR^{y} (\hat \beta)$ has the same RLE as $\sigma_j$, but begins with 1. 
Thus, since $\bR_j$ is CCR-related and applying Observation~\ref{obs:PRR}, $\PRR^{y+ (n-1)}(\hat \beta) = \sigma_j$, 
and thus $\PRR^{(n-1)+x-y}(\hat \alpha_i)  = \hat \beta$.
%.  
By definition, $\PRR^{y+x}(\gamma_j) = \sigma_j$, which means that $\PRR^{y}(\gamma_j)  = \hat \alpha_i$.   
Since $\bR_j$ is CCR-related, $\alpha_j = \gamma_j$.  Thus $a = (n{-}1)+x-y$, $b = (2n{-}2)-y$, and clearly $a<b$.
For all other cases such that $\beta \neq \sigma_i$ (the opp-rep), it is a simple exercise to see that $R_j < R_k$, and hence $j<k$. 

\begin{exam} \small
Consider $\bR_i$ where $\alpha_i = \sigma_i = 10011001111  \in \SpecialO(11)$ and has RLE 12224.  The corresponding LC2-rep
 $\beta = \gamma_i = 01111001100$ has RLE 14222.   Below are the strings from $\bR_j$
including $\hat \sigma_i$ and $\hat \gamma_i$ in the order that they 
 appear in  $\mathcal{Y}'_{11}$.   Note that $\hat \beta$ appears after $\hat \alpha_i$ ($a=8, b=16$).

\medskip

\begin{tabular} {l}
00000011001 \\
00000110011 \\
00001100111 \\ 
00011001111  \  \ $\leftarrow$  \ \   \red{$\hat \sigma_i = \hat \alpha_i$} \\ 
00110011111 \\
01100111111  \ \ $\leftarrow$   \ \    $\sigma_j$, the RL2-rep, with RLE 1226 \\
11001111110 \\
10011111100 \\
00111111001 \\
01111110011 \\
11111100110  \\
11111001100   \ \ $\leftarrow$  \ \   \red{$\hat \gamma_i = \hat \beta$} \\
11110011000 \\
11100110000  \\
11001100000 \\
10011000000  \ \ $\leftarrow$   \ \    $\overline{\sigma_j}$ \\
00110000001 \\
01100000011 \\
11000000110 \\
10000001100 \ \ $\leftarrow   \ \ \alpha_j = \gamma_j$, the opp-rep and LC2-rep for this cycle \\
\end{tabular}

\end{exam}
%+====================

%==============================
% NOT SPECIAL
%==============================
\smallskip

\noindent
{\bf CASE 2: $\sigma_i  \notin \SpecialO(n)$}.  ~
% m = 1
If $m=1$ then $\bR_i$ is CCR-related and $R_i = 1r_1$.  
Thus $\beta = 01^{n-1} = \sigma_i$ since it is is not an opp-rep.  
However, $\hat \beta = 1^n$ is an opp-rep.  Contradiction.   
% m > 1
Thus, assume $m>1$.
By definition $\alpha_i = \gamma_i$. 
Assume $R_i = 1r_1r_2\cdots r_m$.  Then,
\begin{itemize}
\item $\alpha_i = \gamma_i$ has RLE  $1r_mr_1r_2\cdots r_{m-1}$ where  $\PRR^{r_m}(\alpha_i) = \sigma_i$, and
\item $\hat \alpha_i$ has RLE $(r_m{+}1)r_1r_2\cdots r_{m-1}$ and is in $\bR_j$.
\end{itemize}
Consider the RLE possibilities of the other strings in $\bR_j$ as outlined in Lemma~\ref{lem:RLEs}.  Given that $\sigma_i$ is an RL-rep, clearly
\begin{center}
$\sigma_j = \left\{ \begin{array}{ll}
        \PRR^{r_m}(\hat \alpha_i) &\ \  \mbox{if $\sigma_i$ begins with 0;}\\
        \PRR^{(n-1) + r_m}(\hat \alpha_i)  &\ \  \mbox{if $\sigma_i$ begins with 1 ($\bR_i$ is PCR-related).}\end{array} \right.$
\end{center}
In both cases $\sigma_j$ begins with 0 (implying $\alpha_j = \gamma_j$) and  $R_j =  1r_1r_2\cdots r_{m-2}(r_{m-1}{+}r_m)$.
\begin{claim} \label{claim:parent2}
If $\bR_j$ is the parent of $\bR_i$ then $\sigma_j$ begins with 0 and $\bR_j$ is aperiodic.  
\end{claim} 
Note this claim also held for the case when $\sigma_i  \in \SpecialO(n)$.  Observe that $\bR_j$ is indeed aperiodic, since if we assume otherwise, it implies that $\sigma_i$ is not an RL2-rep.  
By definition of an LC2-rep, $\sigma_j = \PRR^{r_{m-1} + r_m}(\gamma_j).$  
Thus: 
\begin{center}
\begin{equation} \label{eq:oppb}
\alpha_j = \gamma_j = \left\{ \begin{array}{ll}
        \PRR^{(n-1) - r_{m-1}}(\hat \alpha_i) &\ \  \mbox{if $\bR_i$ is CCR-related;}\\
         \PRR^{(2n-2)  - r_{m-1}}(\hat \alpha_i) &\ \  \mbox{if $\bR_i$ is PCR-related and $\sigma_i$ begins with 0;} \\
        \PRR^{(n-1) - r_{m-1}}(\hat \alpha_i) &\ \  \mbox{if $\bR_i$ is PCR-related and $\sigma_i$ begins with 1.}
        \end{array} \right.
 \end{equation}
\end{center}

% \Beta = \overline{\gamma_i}
Suppose $\beta = \overline{\gamma_i}$. 
 If $\bR_i$ is CCR-related then $\sigma_j$ begins with 0, but $\sigma_k$ begins with 1, and hence $j<k$.
Otherwise, $\bR_i$ is PCR-related and  $\bR_j$ is CCR-related and hence $\alpha_j = \gamma_j$.   Since both $\gamma_i$ and $\overline{\gamma_i}$ belong to $\bR_i$,  both $\sigma_i$ and $\overline{\sigma_i}$  belong to $\bR_i$.  Thus $\sigma_i$ begins with 0 and from (\ref{eq:oppb}), $b = (2n{-}2) - r_{m-1}$.
 Since $\alpha_i = \gamma_i$,  we have
$\PRR^{n-1}(\hat \alpha_i) =  \overline{\hat \alpha_i} = \hat \beta$.   
%As noted earlier, $\sigma_j = \PRR^{r_m}(\hat \alpha_i)$  and $\sigma_j = \PRR^{r_{m-1}+r_m}(\gamma_j)$ 
%which implies $\hat \alpha_i = \PRR^{r_{m-1}}(\gamma_j)$.  
Thus, $a=n{-}1$ and clearly $a<b$.

Since the RLE of $\beta$ begins with 1,  from Lemma~\ref{lem:RLEs}, the RLE for $\beta$ must be of
the form $1r_s\cdots r_mr_1\cdots r_{s-1}$ for some $1 \leq s \leq m$.  Similar to our analysis for $R_j$,
$R_k$ must begin with 1 followed by a rotation of $r_{s+1}\cdots r_mr_1\cdots r_{s-2}(r_{s-1}{+}r_s)$.
Suppose $1 < s \leq m$.  Let $r_1\cdots r_m = (r_1\cdots r_p)^q$ for some largest $q \geq 1$.  
Then since $\sigma_i$ is an RL2-rep, $R_j < R_k$ unless $s$ is a multiple of $p$, in which 
case  $\beta = \gamma_i$ (opp-rep) or $\beta = \overline{\gamma_i}$ (already handled).
Suppose $s = 1$, which means $\beta = \sigma_i$ or $\beta = \overline{\sigma_i}$.   Since $\sigma_i$ is an RL2-rep,
for each $1 < s' \leq m$, the string $r_{s'}\cdots r_{m}$ is less than or equal to the prefix of $\sigma_i$ of the same length.
Thus, for $s' \neq 2$, $R_j < R_k$.  If $s'=2$, $R_j < R_k$ unless $r_1\cdots r_{m-2} = r_2\cdots r_{m-1}$ and $r_1 = r_{m-1}$, in which case $R_j = R_k$.  Since $\sigma_i$ is an RL2-rep, $r_m \geq r_1$.
Thus, $\beta$ has RLE of the form $1x^zy$ where $y \leq x$.  If $y=x$, then $\beta= \gamma_i$ (opp-rep) or $\beta = \overline{\gamma_i}$ (already handled).  Thus consider $y > x$.  We now consider whether or not $\bR_i$ is CCR-related or PCR-related.  Note $r_{m-1} = x$ and $r_m = y$.

Suppose $\bR_i$ is CCR-related.  If $\beta = \sigma_i$, then $\sigma_j$ begins with 0, but $\sigma_k$ begins with~1 and thus $j<k$.
Otherwise, if $\beta = \overline{\sigma_i}$ then $j=k$.  From (\ref{eq:oppb}), $b = (n{-}1) -x$.   Note that
$\PRR^{x}(\hat \beta) = \sigma_j$ and previously we observed that $\PRR^{y}(\hat \alpha_i) = \sigma_j$.
Thus, $a = y-x$ and clearly $a<b$.

Suppose $\bR_i$ is PCR-related. By its RLE, clearly $\overline{\sigma_i}$ is not in $\bR_i$.  Thus $\beta = \sigma_i$.
If $\beta$ begins with 1, then since $z$ must be odd, $\beta  \in \SpecialO(n)$ -- contradiction.
If $\beta$ begins with 0, then observe that $\PRR^x(\hat \beta) = \overline{\sigma_j}$ and hence
$\PRR^{(n-1) - x}(\sigma_j) = \hat \beta$.   Thus $a = (n-1)-x+y$.  From (\ref{eq:oppb}), $b=(2n-2)-x$ and
clearly $a<b$.

%

%Otherwise, $R_k$ must have the form
%
%\small
%\begin{eqnarray}  
 %  &&  1r_{s'}\cdots r_{s-2}(r_{s-1}{+}r_s)r_{s+1}\cdots r_mr_1\cdots r_{s'-1},   \mbox{\ \ such that  $0 < s' < s$,  or }  \label{eq:rk1opp} \\   
 %   && 1r_{s'}\cdots r_mr_1\cdots r_{s-2}(r_{s-1}{+}r_s)r_{s+1}\cdots r_{s'-1},   \mbox{\ \ such that  $s < s' \leq m$.  }   \label{eq:rk2opp} 
%\end{eqnarray}
%\normalsize 
%\smallskip

\smallskip

%{\bf Suppose the RLE for $\beta$ is periodic.} Then the RLE for $\beta$ must have the form $1 (r_1\cdots r_p)^q$ for some $p,q > 1$, and we can assume $1 \leq s \leq p$.  

%===========================
%===========================
\subsection{ $j<k$ } \label{sec:jk2}

This section applies the same arguments as Section~\ref{sec:jk}.

%========================
\begin{claim} \label{claim:jk2}
If $\bR_j$ and $\bR_k$ have the same run length where $j<k$ such that $\sigma_j$ and $\sigma_k$  both begin with~0, then every string from $\bR_j$ appears in $\mathcal{Y}'_n$ before any string from $\bR_k$.
\end{claim}
%========================
\begin{proof}
The proof is  by induction on the levels of the related tree of cycles rooted by $\bR_1$.  The base case trivially holds for cycles with run length $1$, as there are not two cycles that meet the conditions.  Assume that the result holds for all cycles at levels with run length less than $\ell > 1$.
Consider two cycles $\bR_j$ and $\bR_k$ with run length $\ell$ such that $\sigma_j$ and $\sigma_k$  both begin with~0; neither $\sigma_j$ nor $\sigma_k$ are opp-special and since $j < k$,   $R_j < R_k$.
 Let $\bR_x$ and $\bR_y$ denote the parents of $\bR_j$ and $\bR_k$, respectively.  By Claim~\ref{claim:parent2}, both $\sigma_x$ and $\sigma_y$ begin with 0.
 Given $R_j > R_k$,  our earlier analysis (just before Claim~\ref{claim:parent2}) implies that the RLE of $\sigma_x$ is less than the RLE of $\sigma_y$.
Thus, by the ordering of the cycles,  $x < y$. 
By induction, every string from $\bR_x$ appears before  every string from $\bR_y$ in $\mathcal{Y}'_n$, and hence by Observation~\ref{obs:tree} (item 4), we have our result.  
\end{proof}

%========================
\begin{claim} \label{claim:same2}
Let $\bR_k$ and $\bR_{k'}$ be cycles with $k'< k$ such that $\sigma_k$ and $\sigma_{k'}$ have the same RLE  $1r_1r_2\cdots r_m$ where $m \geq 1$.  Then every string from $\bR_{k'}$ appears in $\mathcal{Y}'_n$ before any string from $\bR_k$.
\end{claim}
%========================
%
\begin{proof}
By the ordering of the cycles, $\sigma_{k'}$ begins with 0 and $\sigma_k$ begins with 1; they belong to PCR-related cycles.
Note that $\sigma_k = \overline{\sigma_{k'}}$ and similarly $\gamma_k = \overline{\gamma_{k'}}$.  Thus, $\hat \sigma_k = \overline{\hat \sigma_{k'}}$   and $\hat \gamma_k = \overline{\hat \gamma_{k'}}$ and each pair, respectively, will belong to the same CCR-related cycle.  
If $\sigma_k  \in \SpecialO(n)$ we previously observed that $\hat \gamma_k$ and $\hat \sigma_k$ belong to the same cycle, and thus $\bR_k$ and $\bR_{k'}$ have the same parent.
If $\sigma_k  \notin \SpecialO(n)$, then $\bR_k$ and $\bR_{k'}$ also have the same parent containing both $\hat \gamma_k$ and $\hat \gamma_{k'}$.
Let $\bR_{\ell}$ be the shared parent of $\bR_k$ and $\bR_{k'}$.    
Since $\sigma_{k'}$ begins with 0,  $\alpha_{k'} = \gamma_{k'}$.   
If $m=1$ there is only one cycle and it is CCR-related.
If $m>1$, then $\alpha_{\ell} =  \PRR^{(2n-2) - r_{m-1}}(\hat \alpha_{k'})$ from (\ref{eq:oppb}).
If $\sigma_k  \in \SpecialO(n)$, then $\alpha_k = \sigma_k$ and from earlier analysis $\alpha_{\ell} =  \PRR^{(2n-2) - r_m}(\hat \alpha_k)$, where $r_m = y$; otherwise $\alpha_k = \gamma_k$, 
 and $\alpha_{\ell} =  \PRR^{(n-1) - r_{m-1}}(\hat \alpha_k)$.  
In both cases, applying Observation~\ref{obs:tree},  $\hat \alpha_{k'}$ appears before $\hat \alpha_k$ in $\mathcal{Y}'_n$,
and every string in $\bR_{k'}$ appears in $\mathcal{Y}'_n$ before any string in $\bR_k$.
\end{proof}
%=======
  
Recall that $\sigma_j$ begins with 0 from Claim~\ref{claim:parent2}.  Thus, if $\sigma_k$ begins with 0, then Claim~\ref{claim:jk2} implies that all strings from $\bR_j$ appear  in $\mathcal{Y}'_n$ before all strings from $\bR_k$.  
Otherwise, if $\sigma_k$ begins with 1, then it must correspond to a PCR-related cycle.  
Consider $\bR_{k'}$ containing RL2-rep $\overline{\sigma_k}$ which begins with 0; it has the same RLE as $\sigma_k$. 
By Claim~\ref{claim:same2}, all strings from $\bR_{k'}$ appear in $\mathcal{Y}'_n$  before all strings from $\bR_{k}$.  If $j=k'$, we are done; 
otherwise $j<k'< k$ and  Claim~\ref{claim:jk2} implies that all strings from $\bR_j$ appear in $\mathcal{Y}'_n$ before all strings from $\bR_{k'}$.
Finally, by applying Observation~\ref{obs:tree} (item 4), all strings from $\bR_i$ including $\beta$ will appear in $\mathcal{Y}'_n$ before all strings from $\bR_k$ including $\hat \beta$.  

%=========================================================================================
%=================================================================
%=================================================================
\section{Future work} \label{sec:fut}

The following questions provide avenues for future research.   \\ 

\noindent
{\bf P1.}  Can  $\same{n}$, $\opposite{n}$, or $\LC{n}$ be generated via a concatenation approach, and if so, can they be generated in $O(1)$ time per symbol using polynomial space?  \medskip

\noindent
{\bf P2.} The (greedy) prefer-same and prefer-opposite de Bruijn sequences for alphabets of size $k >2$ are described at \url{http://debruijnsequence.org}.
Are there simple de Bruijn successors for these generalized sequences? \medskip

\noindent
{\bf P3.} Does there exist an efficient decoding algorithm for the sequences $\same{n}$, $\opposite{n}$, or $\LC{n}$?   That is, without generating the sequence,
at what position $r$ do we find a given string $\omega$ (unranking)?  And, given a string $\omega$, at what position $r$ does it appear (ranking)?

\medskip

\noindent
{\bf P4.} Answer Conjecture~\ref{conj:lc}.  \medskip

\noindent
{\bf P5.} Can Fredricksen and Kessler's de Bruijn sequence construction $\LC{n}$~\cite{lexcomp} be generalized to larger alphabets?\\

\bigskip

%=========================================================================================
%=========================================================================================
%=========================================================================================

%----------- BIBLIOGRAPHY --------------
%\bibliographystyle{plainurl}
\bibliographystyle{abbrv}

\bibliography{refs.bib}

\begin{thebibliography}{1}

\bibitem{DBLP:journals/cacm/Dijkstra68a}
Edsger~W. Dijkstra.
\newblock Letters to the editor: go to statement considered harmful.
\newblock {\em Commun. {ACM}}, 11(3):147--148, 1968.
\newblock \href {https://doi.org/10.1145/362929.362947}
  {\path{doi:10.1145/362929.362947}}.

\bibitem{DBLP:books/mk/GrayR93}
Jim Gray and Andreas Reuter.
\newblock {\em Transaction Processing: Concepts and Techniques}.
\newblock Morgan Kaufmann, 1993.

\bibitem{DBLP:conf/focs/HopcroftPV75}
{John E.} Hopcroft, {Wolfgang J.} Paul, and {Leslie G.} Valiant.
\newblock On time versus space and related problems.
\newblock In {\em 16th Annual Symposium on Foundations of Computer Science,
  Berkeley, California, USA, October 13-15, 1975}, pages 57--64. {IEEE}
  Computer Society, 1975.
\newblock \href {https://doi.org/10.1109/SFCS.1975.23}
  {\path{doi:10.1109/SFCS.1975.23}}.

\bibitem{DBLP:journals/cacm/Knuth74}
Donald~E. Knuth.
\newblock {Computer Programming as an Art}.
\newblock {\em Commun. {ACM}}, 17(12):667--673, 1974.
\newblock \href {https://doi.org/10.1145/361604.361612}
  {\path{doi:10.1145/361604.361612}}.

\end{thebibliography}


\begin{thebibliography}{10}

\bibitem{pref-opposite}
A.~Alhakim.
\newblock A simple combinatorial algorithm for de {B}ruijn sequences.
\newblock {\em The American Mathematical Monthly}, 117(8):728--732, 2010.

\bibitem{alhakim-span}
A.~Alhakim.
\newblock Spans of preference functions for de {B}ruijn sequences.
\newblock {\em Discrete Applied Mathematics}, 160(7-8):992 -- 998, 2012.

\bibitem{revisit}
A.~Alhakim, E.~Sala, and J.~Sawada.
\newblock Revisiting the prefer-same and prefer-opposite de {B}ruijn sequence
  constructions.
\newblock {\em Theoretical Computer Science}, 852:73--77, 2021.

\bibitem{archeo}
J.~Aycock.
\newblock {\em Retrogame Archeology}.
\newblock Springer International Publishing, 2016.

\bibitem{booth}
K.~S. Booth.
\newblock Lexicographically least circular substrings.
\newblock {\em Inform. Process. Lett.}, 10(4/5):240--242, 1980.

\bibitem{nature}
P.~E.~C. Compeau, P.~A. Pevzner, and G.~Tesler.
\newblock How to apply de {B}ruijn graphs to genome assembly.
\newblock {\em Nature Biotechnology}, 29(11):987--991, 2011.

\bibitem{DB}
N.~G. de~Bruijn.
\newblock A combinatorial problem.
\newblock {\em Indagationes Mathematicae}, 8:461--467, 1946.

\bibitem{grandma2}
P.~B. Dragon, O.~I. Hernandez, J.~Sawada, A.~Williams, and D.~Wong.
\newblock Constructing de {B}ruijn sequences with co-lexicographic order: the
  {$k$}-ary {G}randmama sequence.
\newblock {\em European J. Combin.}, 72:1--11, 2018.

\bibitem{duval}
J.~P. Duval.
\newblock Factorizing words over an ordered alphabet.
\newblock {\em Journal of Algorithms}, 4(4):363--381, 1983.

\bibitem{eldert}
C.~Eldert, H.~Gray, H.~Gurk, and M.~Rubinoff.
\newblock Shifting counters.
\newblock {\em AIEE Trans.}, 77:70--74, 1958.

\bibitem{etzion1987}
T.~Etzion.
\newblock Self-dual sequences.
\newblock {\em Journal of Combinatorial Theory, Series A}, 44(2):288 -- 298,
  1987.

\bibitem{fleury}
M.~Fleury.
\newblock Deux problemes de geometrie de situation.
\newblock {\em Journal de mathematiques elementaires}, 42:257--261, 1883.

\bibitem{flye}
C.~Flye Sainte-Marie.
\newblock Solution to question nr. 48.
\newblock {\em L'interm\'{e}diaire des Math\'{e}maticiens}, 1:107--110, 1894.

\bibitem{fred-succ}
H.~Fredricksen.
\newblock Generation of the {F}ord sequence of length $2^n$, $n$ large.
\newblock {\em J. Combin. Theory Ser. A}, 12(1):153--154, 1972.

\bibitem{fred-nfsr}
H.~Fredricksen.
\newblock A survey of full length nonlinear shift register cycle algorithms.
\newblock {\em Siam Review}, 24(2):195--221, 1982.

\bibitem{lexcomp}
H.~Fredricksen and I.~Kessler.
\newblock Lexicographic compositions and de {B}ruijn sequences.
\newblock {\em J. Combin. Theory Ser. A}, 22(1):17 -- 30, 1977.

\bibitem{fkm2}
H.~Fredricksen and J.~Maiorana.
\newblock Necklaces of beads in $k$ colors and $k$-ary de {B}ruijn sequences.
\newblock {\em Discrete Math.}, 23:207--210, 1978.

\bibitem{gabric-concatenation}
D.~Gabric and J.~Sawada.
\newblock Constructing de {B}ruijn sequences by concatenating smaller universal
  cycles.
\newblock {\em Theoretical Computer Science}, 743:12 -- 22, 2018.

\bibitem{discrep}
D.~Gabric and J.~Sawada.
\newblock Investigating the discrepancy property of de {B}ruijn sequences.
\newblock {\em Submitted manuscript}, 2020.

\bibitem{framework}
D.~Gabric, J.~Sawada, A.~Williams, and D.~Wong.
\newblock A framework for constructing de {B}ruijn sequences via simple
  successor rules.
\newblock {\em Discrete Mathematics}, 341(11):2977 -- 2987, 2018.

\bibitem{karyframework}
D.~{Gabric}, J.~{Sawada}, A.~{Williams}, and D.~{Wong}.
\newblock A successor rule framework for constructing $k$ -ary de {B}ruijn
  sequences and universal cycles.
\newblock {\em IEEE Transactions on Information Theory}, 66(1):679--687, 2020.

\bibitem{golomb}
S.~W. Golomb.
\newblock {\em Shift Register Sequences}.
\newblock Aegean Park Press, Laguna Hills, CA, USA, 1981.

\bibitem{hierholzer}
C.~Hierholzer.
\newblock Deux problemes de geometrie de situation.
\newblock {\em Journal de mathematiques elementaires}, 42:257--261, 1873.

\bibitem{huang}
Y.~Huang.
\newblock A new algorithm for the generation of binary de {B}ruijn sequences.
\newblock {\em J. Algorithms}, 11(1):44--51, 1990.

\bibitem{Jiang2023}
Y.~Jiang.
\newblock A relation between sequences generated by {G}olomb's preference
  algorithm.
\newblock {\em Designs, Codes and Cryptography}, 91(1):285--291, Jan 2023.

\bibitem{cipher}
A.~Klein.
\newblock {\em Stream Ciphers}.
\newblock Springer-Verlag London, 2013.

\bibitem{martin}
M.~H. Martin.
\newblock A problem in arrangements.
\newblock {\em Bull. Amer. Math. Soc.}, 40(12):859--864, 1934.

\bibitem{euler}
P.~A. Pevzner, H.~Tang, and M.~S. Waterman.
\newblock An {E}ulerian path approach to {DNA} fragment assembly.
\newblock {\em Proceedings of the National Academy of Sciences},
  98(17):9748--9753, 2001.

\bibitem{Rubin2017}
A.~Rubin and G.~Weiss.
\newblock Mapping prefer-opposite to prefer-one de {B}ruijn sequences.
\newblock {\em Designs, Codes and Cryptography}, 85(3):547--555, Dec 2017.

\bibitem{sala}
E.~Sala.
\newblock Exploring the greedy constructions of de {B}ruijn sequences.
\newblock Master's thesis, University of Guelph, 2018.

\bibitem{wong}
J.~Sawada, A.~Williams, and D.~Wong.
\newblock A surprisingly simple de {B}ruijn sequence construction.
\newblock {\em Discrete Math.}, 339:127--131, 2016.

\bibitem{williams}
A.~Williams.
\newblock The greedy {G}ray code algorithm.
\newblock In F.~Dehne, R.~Solis-Oba, and J.-R. Sack, editors, {\em Algorithms
  and Data Structures}, pages 525--536, Berlin, Heidelberg, 2013. Springer
  Berlin Heidelberg.

\bibitem{xie}
S.~Xie.
\newblock Notes on de {B}ruijn sequences.
\newblock {\em Discrete Applied Mathematics}, 16(2):157 -- 177, 1987.

\end{thebibliography}

%----------- APPENDIX --------------
%\newpage
%\newgeometry{left=25mm,right=25mm,top=25mm,bottom=25mm}

%\begin{center}
%{\bf Appendix - C code}
%\end{center}
\newpage
\appendix

\section{\large Implementation of the de Bruijn successors  $RL(\omega)$, $LC(\omega)$,  and $S(\omega)$ }

\scriptsize
\begin{code}
#include<stdio.h>
#include<math.h>
#define N_MAX 50
int n;

// =============================================================================
// Compute the RLE of a[1..m] in run[1..r], returning r = ruh length
// =============================================================================
int RLE(int a[], int run[], int m) {
    int i,j,r,old;
    
    old = a[m+1];
    a[m+1] = 1 - a[m];
    r = j = 0;
    for (i=1; i<=m; i++) {
        if (a[i] == a[i+1]) j++;
        else {  run[++r] = j+1;  j = 0;  }
    }
    a[m+1] = old;
    return r;
}
// ===============================================================================
// Check if a[1..n] is a "special" RL representative.  It must be that a[1] = a[n]
// and the RLE of a[1..n] is of the form (21^j)^s1^t where j is even, s >=2, t>=2
// ===============================================================================
int Special(int a[]) {
    int i,j,r,s,t,run[N_MAX];
    
    if (a[1] != 0  || a[n] != 0) return 0;    
    r = RLE(a,run,n);
    
    // Compute j of prefix 21^j
    if (run[1] != 2) return 0;
    j = 0;
    while (run[j+2] == 1 && j+2 <= r) j++;
    
    // Compute s of prefix (21^j)^s
    s = 1;
    while (s <= r/(1+j) -1 && run[s*(j+1)+1] == 2) {
        for (i=1; i<=j; i++) if (run[s*(j+1)+1+i] != 1) return 0;
        s++;
    }
    
    // Test remainder of string is (21^j)^s is 1^t
    for (i=s*(j+1)+1; i<=r; i++) if (run[i] != 1) return 0;
    t = r - s*(1+j);
    
    if (s >= 2 && t >= 2 && j%2 == 0) return 1;
    return 0;
}
// =============================================================================
// Apply PRR^{t+1} to a[1..n] to get b[1..n], where t is the length of the
// prefix before the first 00 or 11 in a[2..n] up to n-2
// =============================================================================
int Shift(int a[], int b[]) {
    int i,t = 0;  
    while (a[t+2] != a[t+3] && t < n-2) t++;
    for (i=1; i<=n; i++) b[i] = a[i];
    for (i=1; i<=n; i++) b[i+n] = (b[i] + b[i+1] + b[n+i-1]) % 2;
    for (i=1; i<=n; i++) b[i] = b[i+t+1];
    return t;
}
// =============================================================================
// Test if b[1..len] is the lex largest rep (under rotation), if so, return the
// period p; otherwise return 0. Eg. (411411, p=3)(44211, p=5) (411412, p=0).
// =============================================================================
int IsLargest(int b[], int len) {
    int i, p=1;
    for (i=2; i<=len; i++) {
        if (b[i-p] < b[i]) return 0;
        if (b[i-p] > b[i]) p = i;
    }
    if (len % p != 0) return 0;
    return p;
}
// =============================================================================
// Membership testers not including the cycle containing 0101010...
// =============================================================================
int RLrep(int a[]) {
    int p,r,rle[N_MAX];
    
    r = RLE(a,rle,n-1);
    p = IsLargest(rle,r);

    // PCR-related cycle
    if (a[1] == a[n]) {
        if (r == n-1 && a[1] == 1) return 0;  // Ignore root a[1..n] = 1010101..
        if (r == 1) return 1;  // Special case: a[1..n] = 000..0 or 111..1
        if (p > 0 && a[1] != a[n-1] && (p == r || a[1] == 1 || p%2 == 0)) return 1;
    }
    // CCR-related cycle
    if (a[1] != a[n]) {
        if (p > 0 && a[1] == 1 && (a[n-1] == 1)) return 1;
    }
    return 0;
}
// =============================================================================
int LCrep(int a[]) {
    int b[N_MAX];
    
    if (a[1] != a[2]) return 0;
    Shift(a,b);
    return RLrep(b);
}
// =============================================================================
int SameRep(int a[]) {
    int b[N_MAX];
    
    Shift(a,b);
    if (Special(a) || (LCrep(a) && !Special(b))) return 1;
    return 0;
}
// =============================================================================
// Repeatedly apply the Prefer-Same or LC or RL successor rule starting with 1^n
// =============================================================================
void DB(int type) {
    int i,j,v,a[N_MAX],REP;

    for (i=1; i<=n; i++) a[i] = 1;  // Initial string
    
    for (j=1; j<=pow(2,n); j++) {
        printf("%d", a[1]);
        
        v = (a[1] + a[2] + a[n]) % 2;
        REP = 0;
        // Membership testing of a[1..n]
        if (type == 1 && SameRep(a)) REP = 1;
        if (type == 2 && LCrep(a)) REP = 1;
        if (type == 3 && RLrep(a)) REP = 1;
        
        
        // Membership testing of conjugate of a[1..n]
        a[1] = 1 - a[1];
        if (type == 1 && SameRep(a)) REP = 1;
        if (type == 2 && LCrep(a)) REP = 1;
        if (type == 3 && RLrep(a)) REP = 1;

        // Shift String and add next bit
        for (i=1; i<n; i++) a[i] = a[i+1];
        if (REP) a[n] = 1 - v;
        else a[n] = v;
    }
}
//------------------------------------------------------
int main() {
    int type;
    
    printf("Enter (1) Prefer-same (2) LC (3) RL: ");  scanf("%d", &type);
    printf("Enter n: ");   scanf("%d", &n);
    
    DB(type);
}
\end{code}

\newpage

\section{\large Implementation of the de Bruijn successors  $RL2(\omega)$, $LC2(\omega)$,  and $O(\omega)$ }

\scriptsize
\begin{code}
#include<stdio.h>
#include<math.h>
#define N_MAX 50
int n;

// =============================================================================
// Compute the RLE of a[s..m] in run[1..r], returning r = run length
// =============================================================================
int RLE(int a[], int run[], int s, int m) {
    int i,j,r,old;
    
    old = a[m+1];
    a[m+1] = 1 - a[m];
    r = j = 0;
    for (i=s; i<=m; i++) {
        if (a[i] == a[i+1]) j++;
        else {  run[++r] = j+1;  j = 0;  }
    }
    a[m+1] = old;
    return r;
}
// ===============================================================================
// Check if a[1..n] is a "special" RL representative: the RLE of a[1..n] is of
// the form 1 x^j y where y > x and j is odd. Eg. 12224, 1111113 (PCR-related)
// ===============================================================================
int Special(int a[]) {
    int i,r,rle[N_MAX];
    
    r = RLE(a,rle,1,n);
    if (r%2 == 0) return 0;
    for (i=3; i<r; i++) if (rle[i] != rle[2]) return 0;
    if  (a[1] == 1 && a[2] == 0 && i == r && rle[r] > rle[2]) return 1;
    return 0;
}
// =============================================================================
// Apply PRR^{t} to a[1..n] to get b[1..n], where t is the length of the
// prefix in a[1..n] before the first 01 or 10 in a[2..n]
// =============================================================================
int Shift(int a[], int b[]) {
    int i,t=1;
    
    while (a[t+1] == a[t+2] && t < n-1) t++;
    for (i=1; i<=n; i++) b[i] = a[i];
    for (i=1; i<=n; i++) b[i+n] = (b[i] + b[i+1] + b[n+i-1]) % 2;
    for (i=1; i<=n; i++) b[i] = b[i+t];
    return t;
}
// =============================================================================
// Test if b[1..len] is the lex smallest rep (under rotation), if so, return the
// period p; otherwise return 0. Eg. (114114, p=3)(11244, p=5)(124114, p=0).
// =============================================================================
int IsSmallest(int b[], int len) {
    int i, p=1;
    for (i=2; i<=len; i++) {
        if (b[i-p] > b[i]) return 0;
        if (b[i-p] < b[i]) p = i;
    }
    if (len % p != 0) return 0;
    return p;
}
// =============================================================================
// Membership testers with case for 111111...1  (run length for a[2..n])
// =============================================================================
int RL2rep(int a[]) {
    int p,r,rle[N_MAX];
    
    r = RLE(a,rle,2,n);
    if (r == 1) return 1;       // Special case: a[1..n] = 000..0 or 111..1
    if (a[1] == a[2]) return 0;
    p = IsSmallest(rle,r);

    if (a[1] == a[n] && p > 0 && (p == r || a[1] == 0 || p%2 == 0)) return 1;  //PCR-related
    if (a[1] != a[n] && p > 0 && a[1] == 0) return 1;  // CCR-related
    return 0;
}
// =============================================================================
int LC2rep(int a[]) {
    int t,b[N_MAX];
    
    if (a[1] == a[2]) return 0;
    t = Shift(a,b);
    return RL2rep(b);
}
// =============================================================================
int OppRep(int a[]) {
    int b[N_MAX];
    
    Shift(a,b);
    if (Special(a) || (LC2rep(a) && !Special(b))) return 1;
    return 0;
}
// =============================================================================
// Repeatedly apply the Prefer Opp or LC or RL successor rule starting with 1^n
// =============================================================================
void DB(int type) {
    int i,j,v,a[N_MAX],REP;

    // Initial string
    for (i=1; i<=n; i+=2) a[i] = 0;
    for (i=2; i<=n; i+=2) a[i] = 1;
    
    for (j=1; j<=pow(2,n); j++) {
        printf("%d", a[1]);
        
        v = (a[1] + a[2] + a[n]) % 2;
        REP = 0;
        // Membership testing of a[1..n]
        if (type == 1 && OppRep(a)) REP = 1;
        if (type == 2 && LC2rep(a)) REP = 1;
        if (type == 3 && RL2rep(a)) REP = 1;
        
        // Membership testing of conjugate of a[1..n]
        a[1] = 1 - a[1];
        if (type == 1 && OppRep(a)) REP = 1;
        if (type == 2 && LC2rep(a)) REP = 1;
        if (type == 3 && RL2rep(a)) REP = 1;

        // Shift String and add next bit
        for (i=1; i<n; i++) a[i] = a[i+1];
        if (REP) a[n] = 1 - v;
        else a[n] = v;
    }
}
//------------------------------------------------------
int main() {
    int type;
    
    printf("Enter (1) Prefer-opposite (2) LC2 (3) RL2: ");  scanf("%d", &type);
    printf("Enter n: ");   scanf("%d", &n);

    DB(type);
}
\end{code}

\end{document}